%% file: paper.tex
\documentclass[10pt, conference]{IEEEtran}
\usepackage[utf8]{inputenc}
\usepackage[T1]{fontenc}

\usepackage{amsmath}
\usepackage{amssymb}
\usepackage{amsfonts}
\usepackage{amsthm}

\usepackage{placeins}
\let\labelindent\relax
\usepackage{enumitem}
\usepackage{mathtools}

\usepackage[american]{babel}
\usepackage{hyperref}
\hypersetup{colorlinks=true, urlcolor=black, linkcolor=black, citecolor=black}

\title{Discrete Load Balancing in Heterogeneous Networks with a Focus on Second-Order Diffusion}

\author{\IEEEauthorblockN{\makebox[\textwidth]{\hfill Hoda Akbari\IEEEauthorrefmark{1} \hfill Petra Berenbrink\IEEEauthorrefmark{1} \hfill Robert Elsässer\IEEEauthorrefmark{2} \hfill Dominik Kaaser\IEEEauthorrefmark{2} \hfill}}
\IEEEauthorblockA{
\makebox[0.5\textwidth]{\IEEEauthorrefmark{1}Simon Fraser University}
\makebox[0.5\textwidth]{\IEEEauthorrefmark{2}University of Salzburg}
}
\IEEEauthorblockA{
\makebox[0.5\textwidth]{\textit{\{hodaa, petra\}@sfu.ca}}
\makebox[0.5\textwidth]{\textit{\{elsa, dominik\}@cosy.sbg.ac.at}}
}
}

\usepackage{balance}

\usepackage[textwidth=0.6in,textsize=tiny]{todonotes}
\input{macros}

\begin{document}

\maketitle
\begin{abstract}
\input{abstract}
\end{abstract}

\input{introduction}

\input{model}

\input{fos}

\input{sos}

\input{negative}

\input{simulation}

\input{conclusion}

\FloatBarrier
\balance
\bstctlcite{BSTcontrol}
\bibliographystyle{IEEEtranS}
\bibliography{paper}


\end{document}

%% file: macros.tex
\newcommand{\shorten}[1]{}

\newcommand*\tageq{\refstepcounter{equation}\tag{\theequation}}

\newcommand{\bx}{\ensuremath{\mathbf{x}}}
\newcommand{\by}{\ensuremath{\mathbf{y}}}

\newcommand{\be}{\ensuremath{\mathbf{w}}}
\newcommand{\bi}{\ensuremath{\hat{\mathbf{i}}}}
\newcommand{\bj}{\ensuremath{\hat{\mathbf{j}}}}
\newcommand{\ba}{\ensuremath{\mathbf{a}}}
\newcommand{\bOne}{\ensuremath{\mathbf{1}}}
\newcommand{\bv}{\ensuremath{\mathbf{v}}}

\newcommand{\bu}{\ensuremath{\mathbf{u}}}
\newcommand{\x}[3]{\ensuremath{x^{{#1}}_{#2}{(#3)}}}
\newcommand{\X}[3]{\ensuremath{X^{{#1}}_{#2}{(#3)}}}
\newcommand{\y}[4]{\ensuremath{y^{{#1}}_{#2,#3}{(#4)}}}
\newcommand{\Y}[4]{\ensuremath{Y^{{#1}}_{#2,#3}{(#4)}}}
\newcommand{\xv}[2]{\ensuremath{x^{{#1}}(#2)}}
\newcommand{\yv}[2]{\ensuremath{y^{{#1}}(#2)}}
\newcommand{\yell}{\mathbb{Y}_{\boldsymbol\ell}}
\newcommand{\yellmin}{\mathbb{Y}_{\boldsymbol\ell-1}}
\newcommand{\e}[3]{\ensuremath{e_{#1,#2}(#3)}}
\newcommand{\E}[3]{\ensuremath{E_{#1,#2}(#3)}}
\newcommand{\Q}[4]{\ensuremath{\mathcal{C}^{#1}_{#2,#3}(#4)}}
\newcommand{\Qv}[3]{\ensuremath{\mathcal{C}^{#1}_{#2}(#3)}}
\newcommand{\DC}[2]{\ensuremath{C^{\infty}\circ D^{#1}}}

\usepackage{aliascnt}
\usepackage{xparse}
\NewDocumentCommand{\xnewtheorem}{m o m}{\IfNoValueTF{#2}{\newtheorem{#1}{#3}\expandafter\def\csname #1autorefname\endcsname{#3}}{
\newaliascnt{#1}{#2}
\newtheorem{#1}[#1]{#3}
\aliascntresetthe{#1}
\expandafter\def\csname #1autorefname\endcsname{#3}
}
}

\newtheorem{theorem}{Theorem}
\xnewtheorem{lemma}[theorem]{Lemma}
\xnewtheorem{observation}{Observation}
\xnewtheorem{definition}{Definition}
\xnewtheorem{remark}{Remark}

\addto\extrasamerican{
}

\newcommand\thmref[1]{\autoref{thm:#1}}
\newcommand\defref[1]{\autoref{def:#1}}
\newcommand\secref[1]{\autoref{sec:#1}}
\newcommand\lemref[1]{\autoref{lem:#1}}
\newcommand\obsref[1]{\autoref{obs:#1}}

\newcommand\appref[1]{Appendix \ref{app:#1}}
\renewcommand\appref[1]{the appendix}
\newcommand\eq[1]{\eqref{eq:#1}}

\DeclareMathOperator{\PROBABILITY}{P\scriptstyle{r}}

\DeclareMathOperator{\EXPECTED}{E}

\let\Pro\Probability
\let\PrBig\Probability

\let\Exbig\Expected
\let\ExBig\Expected
\let\Ex\Expected

%% file: abstract.tex
In this paper we consider a wide class of discrete diffusion load balancing
algorithms. The problem is defined as follows. We are given an interconnection
network and a number of load items, which are arbitrarily distributed among the
nodes of the network. The goal is to redistribute the load in iterative
discrete steps such that at the end each node has (almost) the same number of
items. In diffusion load balancing nodes are only allowed to balance their load
with their direct neighbors.

We show three main results. Firstly, we present a general framework for
randomly rounding the flow generated by continuous diffusion schemes over the
edges of a graph in order to obtain corresponding discrete schemes. Compared to
the results of Rabani, Sinclair, and Wanka, FOCS'98, which are only valid
w.r.t.~the class of homogeneous first order schemes, our framework can be used
to analyze a larger class of diffusion algorithms, such as algorithms for
heterogeneous networks and second order schemes. Secondly, we bound the
deviation between randomized second order schemes and their continuous
counterparts. Finally, we provide a bound for the minimum initial load in a
network that is sufficient to prevent the occurrence of negative load at a node
during the execution of second order diffusion schemes.

Our theoretical results are complemented with extensive simulations on
different graph classes. We show empirically that second order schemes, which
are usually much faster than first order schemes, will not balance the load
completely on a number of networks within reasonable time. However, the maximum
load difference at the end seems to be bounded by a constant value, which can
be further decreased if first order scheme is applied once this value is
achieved by second order scheme.

%% file: introduction.tex
\section{Introduction}
\label{sect:introduction}

Load balancing is a fundamental task in many parallel and distributed
applications. Often there are significant differences in the amount of work load generated
on the processors of a parallel machine, which have to be balanced in order to
obtain a substantial benefit w.r.t.~the runtime of a parallel computation. One
of the most prominent examples are so-called finite element simulations
\cite{FWM94}.
 
In the load balancing problem we are given an
interconnection network and a number of load items which are arbitrarily
distributed over the nodes of the network. The goal is to redistribute the
items such that at the end each node has (almost) the same  load.
To achieve this goal, nodes are only allowed to communicate with their direct
neighbors. We assume that each node has access to a global clock, and the
algorithm works in synchronous rounds.


A prominent class of load balancing algorithms are so-called diffusion schemes
\cite{DFM99}. In these algorithms, the nodes are allowed to balance their load
with all their neighbors simultaneously in a round. We distinguish between
\emph{continuous} and \emph{discrete} settings. In the continuous case it is
assumed that the load can be split into arbitrarily small pieces. Although
often not realistic, this assumption is very helpful for analyzing these
algorithms \cite{DFM99}. Discrete load balancing algorithms, on the other hand,
assume that tasks are atomic units of load, called \emph{tokens}. Hence, two
adjacent nodes cannot balance their load any way they want; only integral
amounts of load can be transferred. As a consequence, discrete diffusion
algorithms are usually not able to balance the load completely
\cite{ABS12,EMS06}.

Two fundamental diffusion type algorithms are the first order scheme (FOS) and
the second order scheme (SOS) \cite{MGS98}. In the first order scheme the
amount of load that nodes send to their neighbors in a step only depends on
their current load difference. In SOS the flow over an edge is a function of
the current load difference between its incident nodes and the load that was
sent in the previous round. Note that SOS can lead to negative load at some nodes if the
loads of the nodes are not sufficient to fulfill the calculated demand of all
edges. There are tight bounds on the worst-case convergence time of both, FOS
and SOS, in the continuous case \cite{DFM99}. In general, for the optimal choice of
parameters SOS converges much faster than FOS.

The common approach for analyzing discrete diffusion algorithms is to consider
a closely related continuous version of the algorithm and to bound the load
deviation between load vectors of the two processes (\cite{RSW98}).
To explain the approach we need a couple of  definitions first.
We assume that the network is modeled by an undirected graph $G = (V,E)$, where
$V=\{1,\dots,n\}$ represents the set of processors and the edges in $E$
describe the  connections between them. A total
of $m$ identical load items are distributed over the nodes. We use a vector
$x=(x_1,\dots,x_n)$ to indicate the amount of load assigned to every node. In
the \emph{heterogeneous} network model the nodes may have different speeds
$(s_1,\dots,s_n)$. The aim of a load balancing algorithms is to distribute the
load proportional to the processors' speeds. Hence, the ideal load of a node
$i$ is $\bar{x}_i = m s_i/s$, where $s = \sum_{i=1}^n s_i$.
The
deviation of a load vector $x$ from another load vector $x'$ is $\max_{i\le n}
|x_i-x'_i|$. 

In the case of the common approach mentioned above the continuous process would forward a fractional amount of
load $\ell_e$ over some edge $e$, the discrete algorithm rounds $\ell_e$ to an
integer $\ell'_e$. The rounding can be done deterministically or randomized,
whereas randomized rounding often outperforms deterministic rounding (for
example, the always round down approach \cite{SS12}). The difference between $\ell_e$ and
$\ell'_e$ is called the \emph{rounding error}. The propagation of the rounding
errors causes the two processes to deviate from each other.

In this paper we show three main results. Firstly, we present a general
framework for randomly rounding continuous diffusion schemes to discrete
schemes. Compared to the results of \cite{RSW98}, which are only valid
w.r.t.~the class of homogeneous first order schemes, our framework can be used
to analyze a larger class of diffusion algorithms, such as algorithms for
heterogeneous networks and second order schemes. Secondly, we bound the
deviation between randomized second order schemes and their continuous
counterparts. Finally, we provide a bound for the minimum initial load in a
network that is sufficient to prevent the occurrence of negative load at a node
during the execution of second order diffusion schemes.
Our results are supported by extensive simulations on various graph classes,
comparing the performance of FOS and SOS and giving an empiric insight into the
behavior of diffusion based load balancing processes.

%% file: model.tex
\section{Models and Results}\label{sec:known}


\paragraph{First Order Diffusion}
\shorten{
The first order scheme (FOS) was independently introduced by~\cite{C89}
and~\cite{B90}.
}
FOS in the \emph{homogeneous} network model is defined as follows. Let $N(i)$ be the set of neighbors of node $i$ and
$d_i$ be its degree. We define $\xv{}{t}=(\x{}{1}{t}, \dots, \x{}{n}{t})$ as
the \emph{load vector} at the beginning of round~$t\ge 0$, where $\x{}{i}{t}$ is
the load of node $i$. The
amount of load transferred from node~$i$ to node~$j$ in round~$t$ is denoted by $\y{}{i}{j}{t}$. Then FOS is
characterized by the following equations, where $\alpha_{i,j}$ is a parameter, usually $\alpha_{i,j}=1/\left(\max(d_i, d_j)+1\right)$.
\begin{align}
\y{}{i}{j}{t} &= \alpha_{i,j}\cdot \left(\x{}{i}{t} - \x{}{j}{t}\right) \label{eq:fosy}\\
\x{}{i}{t+1} &= \x{}{i}{t} - \sum_{j\in N(i)} \alpha_{i,j}\left(\x{}{i}{t} - \x{}{j}{t}\right) \notag
\end{align}
The process can be expressed with a \emph{diffusion matrix} $M$, where
$M_{i,i}=1-\sum_j \alpha_{i,j}$ and $ M_{i,j}= \alpha_{i,j}$ for $j\in N(i)$.
All other entries of $M$ are zero. Then
\begin{equation}
\xv{}{t+1} = M \cdot \xv{}{t} \enspace , \label{eq:MFOS}
\end{equation}
where $M$ is a symmetric doubly stochastic $n \times n$ matrix.
\shorten{ that can be viewed as the 
transition matrix of an ergodic Markov chain with uniform steady-state distribution.
Hence, repeatedly applying the equation leads to the perfectly balanced state.}
Let $K$ denote the difference between the maximum and minimum load at the beginning 
of the process. Let $\lambda$ denote the second-largest eigenvalue (in magnitude) of $M$.
Then \cite{MGS98, RSW98} show that FOS converges in
$ O\left(\log (Kn)/(1-\lambda)\right)$ rounds.
In \cite{RSW98}
the authors introduce a framework to analyze a wide class of discrete FOS processes. 
This framework served as a foundation for analyzing several discrete FOS 
algorithms. Many of these publications consider uniform processors~\cite{BCFFS11,BFH09,EM03,ES10,FGS12,FS09,GM96,MGS98,RSW98,SS12}, while a few others incorporate processor speeds into the model~\cite{AB12,EMS06}.
The authors of \cite{FGS12} consider a discrete process where the continuous flow is 
rounded randomly. 
This algorithm achieves a deviation bound of 
$O((d\log\log n)/(1-\lambda))$.
The drawback of this method is that rounding up on too many edges might result 
in negative load.
The process of~\cite{BCFFS11} avoids negative load. 
A node first rounds down all the flows on the adjacent edges, which leaves some 
surplus tokens which are randomly distributed among the neighbors.
This algorithm achieves a deviation bound of 
$O(d\sqrt{\log n}+\sqrt{(d\log n \log d)/(1-\lambda)})$.
In \cite{SS12} the authors study two natural discrete diffusion-based protocols 
and their discrepancy bounds depend only polynomially on 
the maximum degree of the graph and logarithmically on $n$.

The balancing process of~\cite{ABS12} 
 simulates a continuous process using a corresponding discrete process. 
In every round the discrete flow on each edge is 
determined such that it stays as close as possible to the total continuous 
flow that is sent over the edge. This process results in a deviation of 
$O(d)$ (for a more detailed description see next section). In \cite{ES10} the authors consider an approach that is based on random 
walks where tokens of overloaded nodes use a random walk to 
reach underloaded nodes. While this approach leads to a situation at the end, in which
no node has more than a constant number of tokens above average \cite{ES10}, 
it needs to keep track of the load traffic the continuous scheme would produce. 
Moreover, 
the corresponding random walks of the tokens 
result in a huge amount of load transmissions between the nodes, which is not 
the case 
in diffusion based schemes \cite{DFM99}.

\paragraph{Second Order Diffusion (SOS)}
Muthukrishnan et~al.~\cite{MGS98} introduce the continuous second order scheme which
is based on a numerical iterative method called successive over-relaxation 
\cite{GV61} and is one of the fastest diffusion load balancing algorithms.
In SOS, the amount of load transmitted over each edge depends on the current load as well as the load transferred in the previous round.
The only exception is the very first round in which FOS is applied.
%
Subsequent rounds follow the equations below.
\begin{align}
\y{}{i}{j}{t} & = (\beta-1)\,\,\y{}{i}{j}{t-1} + \beta \alpha_{i,j} \left(\x{}{i}{t}-\x{}{j}{t}\right) \label{eq:sosy} \\
\x{}{i}{t+1} & = \beta \cdot\left(\x{}{i}{t} - \sum_{j\in N(i)} \alpha_{i,j}\left(\x{}{i}{t} - \x{}{j}{t}\right)\right) \notag \\
& \phantom{{}={}} + (1-\beta)\cdot \x{}{i}{t-1}{} \notag
\end{align}
Here, $\beta$ is independent of the iteration number $t$. 
From the above equations we get
\begin{equation}\label{eq:sositeration}
\xv{}{t+1} =
\begin{cases}
		 M\, \xv{}{t} & \text{if } t=0 \\
		 \beta\cdot M \,\xv{}{t} + (1-\beta)\cdot \xv{}{t-1} & \text{if } t > 0
	\end{cases}
\end{equation}

For the process to converge, $\beta$ must be in the interval $(0,2)$. 
For the optimal choice of $\beta_{\small{opt}} = 2/(1+\sqrt{1-\lambda^2})$
SOS converges in $ O(\log (Kn)/\sqrt{1-\lambda})$ rounds~\cite{MGS98}
which is in general faster than FOS; for graphs with some eigenvalue gap 
$(1-\lambda)^{-1} = \log^{\omega(1)} n$, the convergence time of SOS is 
almost quadratically faster than FOS. 
 Unfortunately, it can happen that the total outgoing flow from a node 
 exceeds its current load, which results in so-called \emph{negative load}.

\paragraph{Heterogeneous Networks}
Continuous FOS and SOS processes in the \emph{heterogeneous network model} 
were first studied in \cite{EMP02}.
In heterogeneous networks, processors have different speeds and the aim is to 
distribute the load proportional to their speeds.
The minimum speed is $1$, the maximum speed is $s_{\max}$, and $s=s_1+\cdots+s_n$.
Let the diagonal matrix $S$ be defined by $S_{i,i} = s_i$.
Then the heterogeneous FOS/SOS processes are defined as before (see (\ref{eq:MFOS}) and (\ref{eq:sositeration})), 
 except the diffusion matrix is now 
$M = I-LS^{-1}$ where $L$ is the normalized Laplacian matrix of the graph \cite{EMP02}.
In \cite{AB12}, the authors analyze a discrete FOS for homogeneous networks.
In \cite{EMP02} the authors
 show that continuous 
FOS/SOS processes converge in $O(\log(Kns_{\max})/({1-\lambda}))$ and 
$O(\log(Kns_{\max})/\sqrt{1-\lambda})$ rounds, respectively.
In \cite{EMS06}, the authors consider a discrete version of 
SOS too. They show that the euclidean distance between the discrete and continuous load vectors in the discrete version is~$O(d\cdot \sqrt{n \cdot s_{\max}}/(1-\lambda))$.


\newcommand\resultbox[1]{\noindent\parbox[t]{0.66in}{\emph{Result #1.}}}

\subsection{New Results}\label{sec:results}
\resultbox{I} We present a general framework for rounding continuous diffusion 
schemes to discrete schemes. 
Our approach described in \secref{general} estimates the error between a continuous
diffusion scheme and the rounded discrete version first, similar to \cite{RSW98}.
Then we combine that error term with 
martingales techniques (similar to the ones used in \cite{BCFFS11}) 
to bound the deviation between the continuous 
scheme and a discrete scheme based on randomized rounding.
Note that the results in \cite{RSW98} are only valid for a class of
homogeneous first order schemes and \cite{BCFFS11} analyzes a fixed 
first order diffusion scheme with a specific transition matrix. In this paper we introduce an error estimation 
that allows us to show results for a larger class of diffusion algorithms
(see Definition \ref{def:linear}) in heterogeneous networks, including SOS.

In the homogeneous case our bounds are the same as the best results for FOS.
Our bound is worse than the $O(d)$ bound of~\cite{ABS12}.
In the current paper 
we bound the deviation of a class of very natural and stateless algorithms. 
That is, the 
amount of load that is forwarded over an edge in step $t$ only depends on the 
load at the beginning of step $t$ and the amount that was sent in step $t-1$. 
The approach of~\cite{ABS12} is not stateless as it simulates the
continuous process. The flow that is sent over the edges in step $t$ takes into 
account the 
difference of the cumulative load that was sent by the continuous process up 
to step $t$ and the cumulative load that was sent by the discrete process so 
far. 

\resultbox{II} We show that randomized SOS has a deviation (after the balancing time of continuous SOS) of
$O\left( d\cdot \log s_{\max}\cdot\sqrt{\log n}/(1-\lambda)^{3/4} \right)$,
where $\lambda$ is the second largest eigenvalue of $M$ and
$s_{\max}$ is the maximum speed.
Note that the runtime of SOS
is in most cases much better than the runtime of FOS, i.e., $ O(\log (Kn)/\sqrt{1-\lambda})$
(assuming optimal $\beta$) compared to
$ O\left(\log (Kn)/(1-\lambda)\right)$ in the case of FOS.

\resultbox{III} We show that the continuous second order scheme with optimal $\beta$
will not generate negative load if at time $t=0$ the minimum load of every node
is at least $O\left(\sqrt{n}\cdot \Delta(0)/\sqrt{1-\lambda}\right)$. Here
$\Delta(0)$ is the difference between the maximum load and the average load at
time $t=0$. For discrete SOS and graphs with proper eigenvalue gap we show a
bound of $O\left((\sqrt{n}\cdot \Delta(0) +d^2)/\sqrt{1-\lambda}\right)$. To
the best of our knowledge these are the first results specifying a sufficient
amount of minimum load w.r.t.~SOS to avoid negative load.

\noindent\parbox[t]{0.66in}{\emph{Simulations.}} We implemented a network and
simulated both, FOS and SOS load balancing processes. Especially in tori, our
results show a clear advantage of SOS over FOS w.r.t.~the number of steps
required to balance the loads. We also empirically analyze the remaining
imbalance that arises in discrete load balancing schemes once the system has
converged such that no node has more than a constant number of additional load
tokens. We propose to switch from SOS to FOS once this threshold is reached,
and our simulations show that this change of the scheme leads to a further drop
of the remaining load imbalance.

%% file: fos.tex
\section{General Framework for FOS Schemes} \label{sec:general}

In this section we first generalize the framework of Rabani et~al.~\cite{RSW98}
to a wider class of processes (see \secref{framework}) and obtain an equation
estimating the deviation of the discrete process from its continuous version.
The estimation is valid as long as the continuous process is \emph{linear}
(\defref{linear}). In \cite{RSW98} the deviation is expressed in terms of the
diffusion matrix. Here, we present an analysis from a different perspective
which allows us to obtain essentially the same deviation formula for a
larger class of processes. Our analysis can be applied to the second order
processes and heterogeneous models. In \secref{rand} we present the framework
that transforms a continuous load balancing process $C$ into a discrete process
$R(C)$ using randomized rounding.

For simplicity we consider in this section only first order processes. In the
next section we generalize the framework to SOS.


\subsection{Deviation between Continuous and Discrete FOS}
\label{sec:framework}

For a load balancing process $A$, we use $\x{A}{i}{t}$ to denote the load of a
node $i$ at the beginning of the round~$t$, and
$\xv{A}{t} = (\x{A}{1}{t},\dots,\x{A}{n}{t})$. For $j \in N(i)$, we define
$\y{A}{i}{j}{t}$ as the amount of load sent from $i$ to $j$ in round $t$ (this
value is negative if load items are transferred from $j$ to $i$), where $N(i)$
represents the set of neighbors of $i$. Then $\yv{A}{t}$ is the matrix with
$\y{A}{i}{j}{t}$ as its entry in row $i$ and column $j$. Note that each
balancing process $A$ can be regarded as a function that, given the current
state of the network, determines for every edge $e$ and round $t$ the amount of
load that has to be transferred over $e$ in $t$. Hence, we can regard
$\yv{A}{t}$ as the result of applying a function $A$, i.e.,
$\yv{A}{t} = A(\xv{A}{t})$. Using this we formally define discrete processes as follows.
\begin{definition}\label{def:rounding}
Let $C$ be a continuous process. A process $D$ is said to be a discrete version
of $C$ with {rounding scheme} $R_D$ if for every vector $\bx$, we have
$D(\bx) = R_D(C(\bx))$ where $R_D$ is a function that rounds each entry of the
matrix to an integer.
\end{definition}

Note that for a load balancing process load conservation over each edge must hold.
Although it may not be a necessary condition, our analyses in this section
requires the process to exhibit a \emph{linearity} property in the following
sense.

\begin{definition}[Linearity]\label{def:linear}
A diffusion process $A$ is said to be \emph{linear} if for all
$\bx, \bx'\in \mathbb{R}^n$ and $a,b\in\mathbb{R}$ we have
$A(a\bx+b\bx') = a \cdot A(\bx)+b \cdot A(\bx')$.
\end{definition}

\begin{lemma}\label{lem:linear}
Both  FOS and SOS  as defined in \secref{known} are  linear.
\end{lemma}
\begin{proof}
Let $M$ be the diffusion matrix and $0\le \beta \le 2 $.
Observe that both FOS and SOS can be described by the following general equation (see equations \eqref{eq:fosy} and \eqref{eq:sosy}).
\begin{align*}
\y{ }{i}{j}{t} &= (\beta-1)\cdot \y{ }{i}{j}{t-1} + \beta\cdot M_{i,j}\cdot \x{ }{i}{t}
&\text{ for } t\ge1,
\end{align*}
Thus  the algorithm $A$ -- where based on the choice of parameter $\beta$, $A$ can represent either FOS and SOS -- is defined by
\begin{equation*} A(\bx,\by) =  (\beta-1)\by + \beta M\bx \end{equation*}
Let $\bx, \bx'\in \mathbb{R}^n, \by,\by'\in \mathbb{R}^{n\times n} $ and $a,b\in\mathbb{R}$.
Then we have
\begin{align*}
\MoveEqLeft A(a\bx+b\bx', a\by+b\by' ) \\
& = (\beta-1)\,(a\by+b\by' ) + \beta M(a\bx+b\bx')\\
& = a((\beta-1) \by+\beta M\bx) + b((\beta-1) \by'+\beta M\bx')\\
&= a A(\bx, \by)+b A(\bx', \by')
\end{align*}
which shows that $A$ is linear.
\end{proof}

Let $C$ be a continuous process and $D$ its discrete version.
Let $\hat{Y}(t)$ represent $C(\xv{D}{t})$.
Then we can say that $D$ always \emph{attempts to} set $\y{D}{i}{j}{t}$
to $\hat{Y}_{i,j}(t)$. Hence, we call $\hat{Y}(t)$ the \emph{continuous
scheduled load}. We define the \emph{rounding error} as
$\e{i}{j}{t} = \hat{Y}_{i,j}(t) -\y{D}{i}{j}{t};$
note that $\e{i}{j}{t} = -\e{j}{i}{t}$.

In the next definition $\bi$ denotes the unit vector of length $n$ with $1$ as its $i$'th entry.
\begin{definition}[Contributions]\label{def:contribution}
Let $\bx$ and \bx' be the load vectors obtained from applying $C$ for $t$ rounds
on $\bi$ and $\bj$, respectively. For two fixed nodes $i$ and $k$ and $j\in N(i)$
the \emph{contribution} of
edge $(i,j)$ on node $k$ after $t$ rounds is defined as
\[\Q{C}{k}{i\rightarrow j}{t} = \bx_k - \bx'_k \enspace . \]
\end{definition}

The next theorem provides a general form of the FOS deviation formula of \cite{RSW98} which has served as a basis for analyzing
several discrete FOS processes.

\begin{lemma}\label{lem:deterministic}
Consider a linear diffusion process $C$ and its discrete version
$D$ with an arbitrary rounding scheme. Then, for an arbitrary
node $k$ and round $t$ we have
$$\x{D}{k}{t}-\x{C}{k}{t}
= \sum_{s=1}^{t} \sum_{\{i,j\}\in E} \e{i}{j}{t-s}\, \Q{C}{k}{i\rightarrow j}{s}$$
\end{lemma}
\begin{proof}
Fix a node $k$ and round $t$.
Suppose we \emph{sequentialize} the load balancing actions of the process
by imposing an arbitrary ordering on the edges.
Then,~$t$ rounds in the parallel view is equivalent to $|E|\cdot t$ steps
in the sequentialized view.
In the following, let $\tau =|E|\cdot t$.
With a slight abuse of notation we let $\DC{\ell}{\tau-\ell}$ denote a
 \emph{hybrid} process in which the load balancing actions are determined
 by $D$ in steps $1$ to $\ell$ and by $C$ afterwards, where $0\le \ell \le \tau$.
Observe that $\x{\DC{\tau}{0}}{k}{t} = \x{D}{k}{t}$ and $\x{\DC{0}{}}{k}{t} = \x{C}{k}{t}$.
Thus we can write $\x{D}{k}{t}-\x{C}{k}{t}$ in the form of a telescoping sum as follows.
\begin{align}\label{eq:successive}
\x{D}{k}{t}-\x{C}{k}{t} &= \x{\DC{\tau}{0}}{k}{t}-\x{\DC{0}{\tau}}{k}{t} \notag \\
&= \sum_{\ell=1}^{\tau} \left(\x{\DC{\ell}{\tau-\ell}}{k}{t}-\x{\DC{\ell-1}{\tau-\ell+1}}{k}{t}\right) 
\end{align}

Fix an arbitrary step $\ell$ and let $\{i,j\}$ and $s$ be the edge and the round corresponding to the step~$\ell$.
Both~$\DC{\ell-1}{\tau-\ell+1}$ and~$\DC{\ell}{\tau-\ell}$ start their round $s+1$ with load vectors that are the same except maybe in~$i$ and~$j$.
This happens because $\DC{\ell-1}{\tau-\ell+1}$ forwards $\hat{Y}_{i,j}(s)$ over $\{i,j\}$ while in $\DC{\ell}{\tau-\ell}$ this amount is $\hat{Y}_{i,j}(s)-\e{i}{j}{s}$.
Thus, by the definition of $\Q{C}{k}{i\rightarrow j}{t}$ and using the linearity property of the process we get
\[ \x{\DC{\ell}{\tau-\ell}}{k}{t}-\x{\DC{\ell-1}{\tau-\ell+1}}{k}{t} = \e{i}{j}{s}\, \Q{C}{k}{i\rightarrow j}{t-s} \enspace .\] 
Plugging the above into \eq{successive} and translating the summation index we get
\begin{align*}
\x{D}{k}{t}-\x{C}{k}{t} &= \sum_{s=0}^{t-1} \sum_{\{i,j\}\in E} \e{i}{j}{s}\, \Q{C}{k}{i\rightarrow j}{t-s}\\
& = \sum_{s=1}^{t} \sum_{\{i,j\}\in E} \e{i}{j}{t-s}\, \Q{C}{k}{i\rightarrow j}{s} \enspace .
\end{align*}
\end{proof}


\subsection{Framework for Randomized FOS}\label{sec:rand}

In this section we use \lemref{deterministic} to analyze a
randomized rounding scheme for a general class of continuous load
balancing algorithms. Our technique is based on the results in
\cite{BCFFS11} where the authors analyzed a fixed discrete FOS process
for homogeneous
$d$-regular graphs using randomized rounding.
Their algorithm is based on a continuous process in which every
node sends a $1/(d+1)$-fraction of its load to each neighbor.
Initially, the discrete algorithm rounds $x_i/(d+1)$ down if it is not an integer.
This leaves
$(d+1)\cdot\lfloor x_i/(d+1)\rfloor$ surplus tokens on node $i$, which they
call \emph{excess} tokens. The excess tokens are then distributed by
sending the tokens to neighbors which are uniformly
sampled without replacement.

Here we apply the technique in a much more general way, using \lemref{deterministic} to express the deviation between the randomized and
deterministic algorithm.
We introduce a randomized framework that
converts a general class of continuous processes to their discrete versions using
randomized rounding.
 For $a\in\mathbb{R}$ we use $\{a\}$ to denote $a - \lfloor a \rfloor$.

\noindent
\textbf{The Randomized Rounding Algorithm.}
Fix a node $i$.
Let $\hat{Y}(t)$ represent $C(\xv{D}{t})$. For each edge $e=\{i,j\}$ let the corresponding $\hat{Y}_{i,j}(t)$ be the load that would be sent over
$e$ by the continuous process $C$. The rounding scheme works as follows.
First, it rounds $\hat{Y}_{i,j}(t)$ down for all the edges.
This leaves $r = \sum_{j:\hat{Y}_{i,j}(t)\ge 0} \{\hat{Y}_{i,j}(t)\} $ excess
load on node $i$. Then it takes~$\lceil r\rceil$ additional tokens and
 sends each of them out with a probability of $r/\lceil r\rceil$. With the remaining
 probability the excess tokens remain on node $i$.
The tokens which do not remain on $i$ are sent to a neighbor $j$
with a probability of
$\{\hat{Y}_{i,j}(t)\}/r$.
Let~$Z_{i,j}(t)$ be a counting random variable denoting the number of excess
tokens that~$i$ sends to~$j$ in round~$t$. Then we have
\begin{equation*}
\Y{R}{i}{j}{t} =
	\begin{cases}
		\lfloor\hat{Y}_{i,j}(t)\rfloor + Z_{i,j}(t) & \text{if } \hat{Y}_{i,j}(t) \geq 0 \\
		-\Y{R}{j}{i}{t} & \text{otherwise.}
	\end{cases}
\end{equation*}


The deviation bound is expressed based on the \emph{refined local divergence} $\Upsilon^{C}(G)$ defined below, which is a function of both the algorithm and the graph:
\[ \Upsilon^{C}(G)= \max_{k \in V} \bigg( \sum_{s=0}^{\infty} \sum_{i=1}^{n}
 \max_{j\in N(i)} \big ( \Q{C}{k}{i\rightarrow j}{s} \big)^2 \bigg)^{1/2} \]
$\Upsilon^{C}(G)$ is a generalization of the refined local divergence $\Upsilon(G)$ introduced in \cite{BCFFS11}.
Then we have the following result.

\begin{theorem}\label{thm:randomized}
Let $C$ be a continuous FOS and let $R= R(C)$ be a discrete FOS using our randomized rounding transformation.
In an arbitrary round $t$ we have w.h.p.\footnote{Throughout this paper, w.h.p.\ means with probability at least $1-n^{-\alpha}$ for some constant $\alpha> 0$.}
\[\left|\X{R}{k}{t} - \x{{C}}{k}{t}\right| = O \left( \Upsilon^{C}(G) \cdot \sqrt{d\log n} \,\right) \enspace .\]
\end{theorem}

The proof of \thmref{randomized} relies on the fact that FOS is a
linear process 
and hence the estimation
of \lemref{deterministic} can be used as a basis for the
randomized analysis.
The proof is similar to the proof of \cite{BCFFS11}, the difference is
that we use $\Q{C}{k}{i\rightarrow j}{t}$'s instead of the diffusion matrix. 
%
We begin the proof of Theorem \ref{thm:randomized} with a simple observation.

\begin{observation}\label{obs:randsoserr}
The following statements are true (Recall that $\{a\}$ denotes $a - \lfloor a \rfloor$).
\begin{enumerate}
\item
If $\hat{Y}_{i,j}(t)\ge 0$ then $
 \E{i}{j}{t} =
\{\hat{Y}_{i,j}(t) \} - Z_{i,j}(t);
$
\item
$\Exbig{ \E{i}{j}{t}} = 0$.
\end{enumerate}
\end{observation}

The first statement of \autoref{obs:randsoserr} holds by definition, since $ \E{i}{j}{t} = \hat{Y}_{i,j}(t) -\Y{R}{i}{j}{t}$ while $\Y{R}{i}{j}{t} = \left\lfloor\hat{Y}_{i,j}(t)\right\rfloor + Z_{i,j}(t)$.
For the second statement, first suppose that $\hat{Y}_{i,j}(t)\ge 0$.  note that $Z_{i,j}(t)$ can be expressed as a sum of $\lceil r\rceil$ identically distributed Bernoulli random variables each of which is one with probability $(r/\lceil r\rceil) \cdot\left(\{\hat{Y}_{i,j}(t) \}/r\right) $.
Thus we have $\Exbig{Z_{i,j}(t) } = \{\hat{Y}_{i,j}(t) \} $ and from there (2) follows from (1). In the case $\hat{Y}_{i,j}(t)< 0$, we have $\hat{Y}_{j,i}(t)=-\hat{Y}_{i,j}(t)> 0$. Thus by the first case we have $\Exbig{ \E{j}{i}{t}} = 0$ and therefore $\Exbig{ \E{i}{j}{t}} =-\Exbig{ \E{j}{i}{t}} = 0$.
 Let $f_k $ denote the difference in the load of~$k$ in round~$t$ of  $R$ and $C$.
In the following, we first observe that $f_k$ is zero in expectation, and then show that it is well concentrated around its average.
\begin{observation}
$\Exbig{f_k}=0 $
\end{observation}
\begin{proof}
The statement follows from \lemref{deterministic} and \obsref{randsoserr}.(2) by  the  linearity of expectation.
\end{proof}

\begin{proof}[Proof of \autoref{thm:randomized}]
As in \cite{BCFFS11}, we are going to use the method of averaged bounded differences to obtain concentration results for the random variable $f_k $.
For a fixed initial load vector $X{(0)}$ the function
$f_k$ depends only on the randomly chosen destinations of the excess tokens. There are $t$ rounds, $n$ nodes, and at most $d$ excess tokens per node per round. Similar to \cite{BCFFS11} we describe these random choices by a sequence of $t  n  d$ random variables, $Y_{1}, Y_{2}, \dots, Y_{tnd}$.  For any $\ell$ with $1 \leq \ell \leq t  n d$, let $(s,i,b) $ be such that
$\ell=snd + (i-1)\,d + b$ (note that $(s,i,b)$ is the $\ell$-th largest element in the sequence). Then $Y_{\ell}$ refers to the destination of the $b$\nobreakdash-th excess token of vertex~$i$ in round~$s$ (if there is one). More precisely,
\begin{equation*}
 Y_{\ell} =
 \begin{cases}
   j & \text{if \;
 the $b$\nobreakdash-th excess token of the vertex~$i$ in }\\
   &\text{ round~$s$ is sent to~$j$, and
   }\\
   0 & \text{otherwise.}
 \end{cases}
\end{equation*}

\noindent
Let $\mathbb{Y}_{\boldsymbol{i}}$ denote $Y_i,\dots, Y_1$. To apply the method of averaged bounded differences, we need to bound the difference sequence below.
\begin{align}
  \left| \Exbig{ f_k \, \mid \, \yell } -
 \Exbig{f_k \, \mid \, \yellmin } \right|. \label{eq:maxdiff}
\end{align}
As in \cite{BCFFS11}, we consider a fixed $\ell$ that corresponds to~$(s_1,i_1,b_1)$ in the lexicographic ordering.

To bound \eq{maxdiff}, we write
\begin{align*}
 c_{\ell}&= \big| \Exbig{ f_k \, \mid \, \yell } - \Exbig{f_k \, \mid \, \yellmin } \big|  \\
&\leq \sum_{s=0}^{t} \sum_{\{i,j\}\in E}
  \left| \Exbig{ \E{i}{j}{s} \mid \, \yell } -
 \Exbig{ \E{i}{j}{s} \mid \, \yellmin }  \right|\\
& \phantom{{}={}} \cdot
\left|\Q{C}{k}{i\rightarrow j}{t-s}\right|
\end{align*}
As in \cite{BCFFS11} we split the sum over $s$ into the three parts $1 \leq s < s_1$, $s = s_1$,
and $s_1 < s \leq t$.
In the following we show that the  sums over $s < s_1$ and $s > s_1$ are both  zero while the part $s = s_1$ is upper bounded by
$ 2 \cdot   \max_{j \in N(i_1)} \left|\Q{C}{k}{i\rightarrow j}{t-s}\right|$.

\noindent\textbf{Case $\mathbf{s < s_1}$:}
For every $\{i,j\} \in E$, $\E{i}{j}{s}$ is already determined by $\mathbb{Y}_{\boldsymbol\ell-1}$.
 Hence,
\begin{align}
\label{eq:main:cases:lower}
& \sum_{s=1}^{s_1-1} \sum_{\{i,j\}\in E}
  \big| \Exbig{ \E{i}{j}{s} \mid \, \yell } -
 \Exbig{ \E{i}{j}{s} \mid \,  \yellmin }  \big| \\
&  \cdot \big| \Q{C}{k}{i\rightarrow j}{t-s} \big|   = 0 \enspace . \notag
\end{align}

\noindent\textbf{Case $\mathbf{s = s_1}$:}
In this case, $\yellmin$ determines $\hat{Y}_{i,j}(t)$ and $\E{i}{j}{s}$ is only affected by $Z_{i,j}{(s)}$'s (see \obsref{randsoserr}).
\begin{align*}
   &  \sum_{\{i,j\}\in E}
 \big| \Exbig{ \E{i}{j}{s} \,\big| \, \yell } -
 \Exbig{ \E{i}{j}{s} \mid \,  \yellmin }  \big|\\
 &\phantom{{}={}} \cdot
 \big| \Q{C}{k}{i\rightarrow j}{t-s}\big| \\
&=  \sum_{i=1}^{n}\sum_{j:\hat{Y}_{i,j}(t) \ge 0}
 \Big(\Big| \Exbig{
   \{\hat{Y}_{i,j}(t) \} - Z_{i,j}(t) \mid \, \yell } \\
 & \phantom{{}={}} -
 \Exbig{ \{\hat{Y}_{i,j}(t) \} - Z_{i,j}(t) \mid \,  \yellmin }  \Big|\Big)
 \cdot
   \big| \Q{C}{k}{i\rightarrow j}{t-s} \big|\\
&=  \sum_{\{i,j\}\in E}
 \Big| \Exbig{
 Z_{i,j}^{(s)} \mid \, \yell }   -
   \Exbig{ Z_{i,j}^{(s)} \mid \,  \yellmin }  \Big|\\
& \phantom{{}={}} \cdot
   \big| \Q{C}{k}{i\rightarrow j}{t-s} \big| \\
&=  \sum_{\{i,j\}\in E}
 \big| \Lambda_{i,j}^{(s)} \big|  \cdot
 \big|\Q{C}{k}{i\rightarrow j}{t-s} \big| \notag \\
&\leq \sum_{i=1}^{n}
 \Big(\max_{j\in N(i)} \big| \Q{C}{k}{i\rightarrow j}{t-s} \big|\Big)\cdot
 \sum_{j\in N(i)}
 \big| \Lambda_{i,j}^{(s)} \big| \enspace, \tageq
 \label{eq:newerboundtwo}
\end{align*}
where we used
\[  \Lambda_{i,j}^{(s)} = \ExBig{  Z_{i,j}^{(s)} \mid \, \yell }   -  \ExBig{   Z_{i,j}^{(s)} \mid \, \yellmin } \]
to simplify the notation.

As in \cite{BCFFS11}, to bound \eq{newerboundtwo} we consider  $\sum_{\{i,j\}\in E} \big| \Lambda_{i,j}^{(s)} \big|$ for $i=i_1$ and $i \neq i_1$ separately.

\smallskip

\noindent\textbf{Case 1:} Let $i = i_1$.
For each  $j \in N(i_1)$, define indicator Bernoulli random variables $I_{u,j}, 1 \le u \le d$, where $I_{u,j}$ is one if the $u$'th excess token of $i_1$ in round $s_1$ goes to $j$ and zero otherwise.
Note that $  Z_{i_1,j}^{(s_1)} = \sum_{1\le u\le d} I_{u,j} $.
Let $r = \sum_{j\in N(i_1)} \left\{\hat{Y}_{i_1,j}(s_1)\right\}$ so that $\lceil r \rceil \geq b_1$ be the number of excess tokens of $i_1$ in round~$s_1$.
Clearly,  $r$
and the destinations of the excess tokens considered in the previous rounds,
are already determined by $Y_{1},\dots,Y_{\ell-1}$. 
 The remaining receivers  $Y_{\ell+1},\dots,Y_{\ell+r-b_1}$
are chosen independently from $N(i_1)\cup \{i_1\}$.
Hence, the choice of $Y_{\ell}$ does not affect the distribution of $I_{u,j}$ except for $u = b_1$, and we have
\begin{align*}
\MoveEqLeft \ExBig{  Z_{i_1,j}{(s_1)} \mid \, \yell }   -  \ExBig{   Z_{i_1,j}{(s_1)} \mid \,  \yellmin }\\
&=
 \ExBig{  I_{1,j}+\dots+I_{\lceil r\rceil,j} \mid \, \yell }  \\
&\phantom{{}={}} -  \ExBig{  I_{1,j}+\dots+I_{\lceil r\rceil,j} \mid \,\yellmin }\\
&=  \ExBig{  I_{b_1,j} \mid \, \yell }   -  \ExBig{  I_{b_1,j} \mid \,  \yellmin }
\end{align*}

Let $w \in  N(i_1)\cup \{i_1\}$ be the destination of the $b_1$\nobreakdash-th excess
token of~$i_1$ in round~$s_1$, that is,
$ Y_{\ell} = w $ and hence,
\begin{equation*}
\Lambda_{i_1,w}^{(s_1)} = 1-\{\hat{Y}_{i_1,w}(t)\}/r \enspace .
\end{equation*}
For any $j \in  N(i_1) \setminus \{w\}$ we have
\begin{equation*}
\Lambda_{i_1,w}^{(s_1)} = -\{\hat{Y}_{i_1,j}(t)\}/r \enspace .
\end{equation*}
Hence,
\begin{align}
\!\! \sum_{j\in N(i_1)} \! \big| \Lambda_{i_1,j}^{(s_1)} \big|
&\le
 1-\{\hat{Y}_{i_1,w}(t)\}/\lceil r\rceil \\
& \phantom{{}={}} + \sum_{j\in N(i_1)\setminus \{w\}} \{\hat{Y}_{i_1,j}(t)\}/\lceil r\rceil \notag\\
&\leq 1 + \sum_{j\in N(i_1)} \{\hat{Y}_{i_1,j}(t)\}/\lceil r\rceil \notag\\
&\leq 2, \label{eq:martingalebound}
\end{align}
where the last inequality holds since \[ \sum_{j\in N(i_1)} \{\hat{Y}_{i_1,j}(t)\} = r \le \lceil r \rceil \enspace .\]

\noindent \textbf{Case 2:}$\ i \neq i_1$.
As $\ell$ corresponds to $(s_1,i_1,b_1)$, the random variable $Z_{i,j}{(s_1)}$ is independent of $Y_{\ell}$ when conditioned on $\mathbb{Y}_{\ell-1}$. Hence, similar to \cite{BCFFS11}, we have
\begin{align*}
  \sum_{\{i,j\}\in E} \big| \Lambda_{i,j}^{(s_1)}\big|
 & = \sum_{j: \{i,j\}\in E}  \left| \ExBig{
  Z_{i,j}^{(s)} \mid \, \yell }   -
 \ExBig{   Z_{i,j}^{(s)} \mid \,  \yellmin } \right|\\
 & = 0 \enspace .
\end{align*}

Combining Case 1 and Case 2 we obtain
\begin{align}
   &  \left( \max_{j\in N(i_1)} \big| \Q{C}{k}{i\rightarrow j}{t-s}\big| \right) \sum_{j: \{i_1,j\}\in E}
   \big| \Lambda_{i_1,j}^{(s)} \big| \notag\\&
+ \sum_{i \in V, i \neq i_1} \left( \max_{j \in N(i)} \big| \Q{C}{k}{i\rightarrow j}{t-s} \big| \right) \sum_{\{i,j\}\in E}
   \big| \Lambda_{i,j}^{(s)} \big|   \notag\\
   \leq {} & \max_{j \in N(i_1)} \big| \Q{C}{k}{i\rightarrow j}{t-s} \big| \cdot 2 + 0 \enspace .
   \label{eq:main:cases:middle}
\end{align}

\noindent\textbf{Case $\mathbf{s > s_1}$:}
Let $\widetilde{\ell}$ be the largest integer that corresponds to round $s-1$. Since $s > s_1$, we have $s - 1 \geq s_1$ and therefore $\widetilde{\ell} \geq \ell$. By the choice of $\widetilde{\ell}$, $Y_{\widetilde{\ell}},\dots,Y_1$ determine the load vector at the end of round $s_1$, $X^{(s_1)}$. By \obsref{randsoserr}, we obtain
$ \Exbig{ \E{i}{j}{s} \, \mid \, Y_{\widetilde{\ell}},\dots,Y_1 } = 0$,
and by the law of total expectation,
\begin{align*}
\MoveEqLeft \Exbig{ \E{i}{j}{s} \,\! \mid \,\! \yell } \\
&=  \Ex{ \Exbig{ \E{i}{j}{s} \,\! \mid \,\! Y_{\widetilde{\ell}},\dots,Y_1 } \,\! \mid \, Y_{\ell},Y_{\ell-1},\dots,Y_1 } \notag\\
&= \Ex{ 0  \,\! \mid \,\! Y_{\ell},Y_{\ell-1},\dots,Y_1 } = 0 \enspace .
\end{align*}
With the same arguments,
$ \Exbig{ \E{i}{j}{s} \, \mid \,  \yellmin } = 0 $, and thus
\begin{align}
\label{eq:main:cases:upper}\notag
&\sum_{s=s_1+1}^{t} \sum_{ \{i,j\}\in E}
  \big| \Exbig{ \E{i}{j}{s} \mid \, \yell } -
 \Exbig{ \E{i}{j}{s} \mid \,  \yellmin }  \big|\\
 & \cdot \big| \Q{C}{k}{i\rightarrow j}{t-s} \big| = 0 \enspace .
\end{align}

\medskip
\noindent
This finishes the case distinction.
Combining equations \eqref{eq:main:cases:lower}, \eqref{eq:main:cases:middle}, and \eqref{eq:main:cases:upper} for the three cases $s < s_1$, $s = s_1$, and $s > s_1$,
similar to \cite{BCFFS11} we obtain that for every fixed $1 \leq \ell \leq tnd$,
\begin{align*}
c_{\ell} &= \left| \Exbig{ f_k \, \mid \, \yell } -
 \Exbig{f_k \, \mid \, \yellmin } \right| \notag \\
 & \le \sum_{s=0}^{t} \sum_{\{i,j\}\in E} \left| \Exbig{ \E{i}{j}{s} \, \mid \, \yell } - \Exbig{ \E{i}{j}{s} \, \mid \,  \yellmin } \right|
\\& \phantom{{}={}} \cdot
 \big|  \Q{C}{k}{i_1\rightarrow j}{t-s_1} \big| \notag\\
 & = 0 + \max_{j\in N(i_1)} \big| \Q{C}{k}{i_1\rightarrow j}{t-s_1}\big| \cdot 2 + 0
 \notag\\
& = 2 \cdot \max_{j \in N(i_1)} \big| \Q{C}{k}{i_1\rightarrow j}{t-s_1} \big|.
\end{align*}

\noindent
Now we consider the sum of the error terms.
\begin{align}
\sum_{\ell=1}^{(t+1) n d} (c_{\ell})^2
 &\le
\sum_{s=0}^{t} \sum_{i=1}^{n} \sum_{b=1}^{d}
 \Big ( 2 \, \max_{j \in N(i)}
 \big| \Q{C}{k}{i\rightarrow j}{t-s} \big|  \Big)^2 \notag\\
& = 4 d \, \sum_{s=0}^{t} \sum_{i=1}^{n}
  \max_{j \in N(i)} \big ( \Q{C}{k}{i\rightarrow j}{s} \big)^2 \notag \\
 &\leq 4d \, \max_{k \in V} \bigg(  \sum_{s=0}^{\infty} \sum_{i=1}^{n}
 \max_{j \in N(i)} \big ( \Q{C}{k}{i\rightarrow j}{s} \big)^2 \bigg)
 \notag\\
&= 8d \, \left( \Upsilon^{\operatorname{C}}(G) \right)^2 \enspace . \label{eq:boundeddiff}
\end{align}
So  by Azuma's inequality \cite[p.~68]{DP09} we have for any $\delta \geq 0$,
\begin{equation*}
 \Pro{ |f_k| > \delta }
 \leq 2 \, \exp
 \big(- \delta^2 \big/ \big({2\sum_{\ell=1}^{t n d} (c_{\ell})^2 } \big)\big) \enspace .
\end{equation*}
Hence by choosing
$\delta= \Upsilon^{C}(G) \, \sqrt{ 32 d \, \ln n } $,
the probability above gets smaller than $2\,n^{-2}$. Applying the union bound
we obtain
\begin{equation*}
\Pro{ \exists k \in V \colon |f_k| > \delta } \leq  2 n^{-1} \enspace .
\end{equation*}
 This implies
\begin{equation*}
\PrBig{  \max_{i,j \in [n]} \big| \X{R}{i}{t} - \x{C}{i}{t} \big| \leq  \delta  } \geq 1 - 2n^{-1} \enspace ,
\end{equation*}
which finishes the proof.
\end{proof}

Using \thmref{randomized} we can also obtain concrete results for
randomized FOS processes as stated in the following theorems.
The first result holds for the homogeneous case and a special class of 
algorithms where $\alpha_{i,j}=1/(\gamma d)$ only. 
Recall that $d$ is the maximum degree.
 The same result was already shown in~\cite{SS12}.

\begin{observation}
Assume $s_1=s_2=\ldots= s_n$ and $\alpha_{i,j}=\frac{1}{\gamma d}$.
Let $C$ be a continuous FOS process and let $R= R(C)$ be a discrete FOS
process based
on the rounding algorithm applied on $C$. Then
\begin{enumerate}
\item[(1)] $ \Upsilon^{\operatorname{C}}(G) = O\left(\sqrt{{\gamma d}/{\left(2-2/\gamma\right)}}\right) \enspace .$
\item[(2)] For any round $t$ we have w.h.p.
\[\left|\x{R}{k}{t} - \x{\operatorname{C}}{k}{t}\right| = O\left(\sqrt{\frac{\gamma d}{2-2/\gamma}}\cdot\sqrt{d\log n} \right) \enspace .\]
\end{enumerate}
\end{observation}

In \cite{SS12} the authors applied a potential function in order to estimate 
$ \Upsilon^{C}(G)$. This proof relies heavily on the fact that the transition probabilities
are uniform for all edges, which is not the case for the heterogeneous model  
or the case where the
$\alpha_{i,j}$ depend on $d_i$ and $d_j$ only.
The next result is more general and applies to both of these cases as well.

\begin{theorem}\label{thm:fosrand}

Let $C$ be a continuous FOS process and let $R= R(C)$ be a discrete FOS
process based
on the rounding algorithm applied on $C$. Then
\begin{enumerate}
\item[(1)] $ \Upsilon^{\operatorname{C}}(G) = O\left(\sqrt{{d\cdot \log s_{\max}}/{\left(1-\lambda\right) }}\right) \enspace .$
\item[(2)] For any round $t$ we have w.h.p.
\[\left|\x{R}{k}{t} - \x{\operatorname{C}}{k}{t}\right| = O\left(d\cdot \sqrt{\frac{\log n\cdot \log s_{\max}}{1-\lambda }}\right) \enspace .\]
\end{enumerate}
\end{theorem}

 To show \autoref{thm:fosrand} we first show the following lemma.

\begin{lemma}\label{lem:msecondnorm}
For an arbitrary $1\le k\le n$,  let the  vector $\ba$ be such that $\ba_i =  M^t_{k,i}- \frac {s_k} s $.
Then we have
$$\|\ba\|_2^2 \le 2\, s_{\max}\,\lambda^{2t} $$
\end{lemma}
\begin{proof}
Let $\hat{\ba}=\hat{\mathbf{k}} -  \frac{s_k}{s}\cdot\bOne_n$. Note that $\ba =\hat{\ba}\, M^t$.
Let $\bv_1,\dots,\bv_n$ be the eigenvectors of $M^t$ with eigenvalues $\lambda^t_1,\dots,\lambda^t_n$, and $\lambda^t$ be the second largest eigenvalue.
Using the fact that $M = I-LS^{-1}$, it is not hard to see that $S^{-1}M^t$ is symmetric. Hence, for each right eigenvector $\bv_i$ of $M^t$ there is a left eigenvector $\bu_i = S^{-1}\bv_i$ with the same eigenvalue $\lambda^t_i$ as proved in the following.
\begin{align*}(M^t)^T \bu_i
&= (M^t)^TS^{-1} S \bu_i
= ( S^{-1}M^t)^T \bv_i \\
&=  S^{-1}M^t\bv_i
= \lambda^t_i S^{-1}\bv_i
=  \lambda^t_i\bu_i
\end{align*}
Also, note that $S^{-1}M^tS\bu_i = S^{-1}M^t\bv_i = \lambda^t_i S^{-1}\bv_i = \lambda^t_i \bu_i $.
As a result, $\bu_i$'s are eigenvectors of $S^{-1}M^tS$, which is symmetric because it is the product of  symmetric matrices $S^{-1}M^t$ and $S$.
Therefore, $\bu_1,\dots,\bu_n$ form an orthonormal basis; so  we can write $\hat{\ba} = \sum_{i=1}^{n} c_i \bu_i$.
Now we write
\begin{equation*} \ba = \hat{\ba}\,M^t = \sum_{i=1}^{n}  c_i \bu_i  M^t  =\sum_{i=1}^{n} \lambda^t_i c_i \bu_i \end{equation*}
Therefore,
\begin{equation}
\|\ba\|_2^2  =   \sum_{i=1}^{n}\lambda_i^{2t} c_i^2 \| \bu_i\|_2^2 
\le \lambda^{2t}\sum_{i=1}^{n} c_i^2 \| \bu_i\|_2^2 
= \lambda^{2t} \| \hat{\ba}\|_2^2 \label{eq:foshatba}
\end{equation}
where  
the  inequality  uses the fact that $\hat{\ba} \perp (s_1,\dots,s_n)$ which is the eigenvector corresponding to the largest eigenvalue.
Also, the last equality follows from the fact that $\bu_i$'s form an orthonormal basis.  On the other hand,
\begin{align*}
 \| \hat{\ba}\|_2^2
&\le  n\cdot  \frac{s_k^2}{s^2}+1
 \le \frac{n \,s_k^2}{(n-1+s_k)^2}+1  \\
&   \le  \frac{n\, s_k^2}{2 (n-1) s_k }+1
\le   s_k+1
\le  2 s_{\max}
\end{align*}
Together with \eq{foshatba}, this yields
\begin{equation*} \|\ba\|_2^2\le 2 \,s_{\max}\lambda^{2t} \enspace , \end{equation*}
as required.
\end{proof}

\begin{proof}[Proof of \thmref{fosrand}]
We have
\begin{align}
\left(\Upsilon^{\operatorname{FOS}}(G)\right)^2
&=  \sum_{t=0}^{\infty} \sum_{i=1}^{n}  \max_{j \in N(i)} \big ( M_{k,i}^t - M_{k,j}^t \big)^2 \notag\\
&\le \sum_{t=0}^{\infty}\sum_{i=1}^n \sum_{j\in N(i)}   \left( M_{k,i}^t - M_{k,j}^t \right)^2 \notag\\
&=  \sum_{t=0}^{t_1-1} \sum_{i=1}^n \sum_{j\in N(i)}   \left( M_{k,i}^t - M_{k,j}^t \right)^2
\notag\\
& \phantom{{}={}} + \sum_{t=t_1}^{\infty}\sum_{i=1}^n \sum_{j\in N(i)}
   \left( M_{k,i}^t - M_{k,j}^t \right)^2 \label{eq:breakupsilon}
\end{align}

In the following, we use $\sigma=\beta-1$ for brevity. Note that $0<\sigma<1$.
\begin{align}
\MoveEqLeft \sum_{t=0}^{t_1-1} \sum_{i=1}^n \sum_{j\in N(i)} \notag
 \left(M_{k,i}^t - M_{k,j}^t \right)^2\\
 &\le \sum_{t=0}^{t_1} \sum_{i=1}^n \sum_{j\in N(i)} 2 \left( (M_{k,i}^t)^2+ (M_{k,j}^t )^2 \right) \notag\\
 &\le 4 \cdot d\cdot \sum_{t=0}^{t_1-1} \sum_{i=1}^n (M_{k,i}^t)^2 \notag\\
 &\le 4 \cdot d\cdot \sum_{t=0}^{t_1-1} \|M^t \,\hat{\textbf{k}} \|_2^2 \notag\\
 &= 4 \cdot d\cdot \sum_{t=0}^{t_1-1}\left( \|\hat{\textbf{k}} \|_2\cdot \max_i \lambda_i^t\right)^2\notag\\
 &= 4 \cdot d\cdot \sum_{t=0}^{t_1-1}1\le
 4 \cdot d\cdot t_1 \label{eq:fosupshead}
\end{align}

Let $t_1 = (\log s_{\max})/(2-2\lambda) $. Note that $\lambda^{1/(1-\lambda)} \le 1/e$. Then we have
\begin{align}
\MoveEqLeft \sum_{t=t_1 }^{\infty}\sum_{i=1}^n \sum_{j\in N(i)}
\left( M_{k,i}^t - M_{k,j}^t \right)^2\notag\\
&\le  \sum_{t=t_1 }^\infty \sum_{i=1}^n \sum_{j\in N(i)} 2\left(\left(M_{k,i}^t- \frac {s_k} s\right)^2+\left(M_{k,j}^t- \frac {s_k} s\right)^2\right)\notag\\
&= 4\cdot d\cdot \sum_{t=t_1 }^\infty \sum_{i=1}^n \left(M_{k,i}^t- \frac {s_k} s\right)^2\notag\\
&\le 8\cdot d\cdot s_{\max}\cdot \sum_{t=t_1 }^\infty \lambda^{2t} \label{eq:msecondnorm}\\
&\le 8\cdot d\cdot s_{\max}\cdot\lambda^{2t_1}\cdot\frac{1}{1-\lambda}
\le \frac{8\,d}{1-\lambda} \label{eq:fosupstail}
\end{align}
where \eq{msecondnorm} follows from \lemref{msecondnorm}.
Combining equations \eqref{eq:breakupsilon}, \eqref{eq:fosupstail}, and \eqref{eq:fosupshead} we get
\begin{equation}
 \left(\Upsilon^{\operatorname{FOS}}(G)\right)^2 = O\left(\frac{d\cdot \log s_{\max}}{1-\lambda }\right) \enspace ,
\end{equation}
which proves the first statement. The bound in the second statement follows immediately from statement (1) and \thmref{randomized}.
\end{proof}

%% file: sos.tex
\section{Second Order Diffusion Processes}\label{sec:sos}
In this section  we show that after some slight adjustments the framework 
of \secref{general} can be applied 
to second order processes on heterogeneous networks.  
All we have to do is to  state definitions \ref{def:linear} and \ref{def:contribution} in a more general 
way that captures the dependence of SOS on the load transfer of the previous round. 
It is easy to see that  \lemref{deterministic} and \thmref{randomized}  
still hold assuming the new definitions. (Note that 
SOS is  linear).
 If $C$ is a second order process, then $\yv{C}{t}$ is determined based on 
 $\xv{C}{t}$  and $\yv{C}{t-1}$. More formally, $\yv{C}{t} = C(\xv{C}{t},\yv{C}{t-1})$.
Thus,  the new definitions also incorporate  $\yv{C}{t-1}$. 
We again use  $\bi$ to denote the unit vector with $1$ as its $i$'th entry.

\begin{definition}[Linearity]\label{def:linear2}
A process $A$ is said to be \emph{linear} if for all $\bx, \bx'\in \mathbb{R}^n, \by,\by'\in \mathbb{R}^{n\times n} $ and $a,b\in\mathbb{R}$ we have $A(a\bx+b\bx', a\by+b\by' ) = a A(\bx, \by)+b A(\bx', \by')$.
\end{definition}

\begin{definition}[Contributions]\label{def:contribution2}
Let $\bx(0)=\bx'(0) = \bi$,  $\by(0) = \mathbf{0}_{n\times n}$ and let $\by'(0)$ be also all zero except $\by'_{i,j}(0)=1$, so that $\bx(1)=\bi$, $\bx'(1) = \bj$. Let $\bx(t+1)$ and $\bx'(t+1)$ be the load vectors obtained from applying $C$ for $t$ rounds on  $(\bx(1), \by(0))$ and $(\bx'(1), \by'(0))$, respectively. Then the \emph{contribution} of the edge $(i,j)$ on a  node $k$ after $t$ rounds is defined as
$\Q{C}{k}{i\rightarrow j}{t} = \bx_k(t) - \bx'_k(t)$.
\end{definition}

To prove  bounds of the deviation of theorems \ref{thm:sosdet} and \ref{thm:sosrand} we 
apply
Observation \ref{obs:randomized2} which follows from
\lemref{deterministic} and \thmref{randomized}. 
This gives us an upper bound in terms of the $\Q{C}{k}{i\rightarrow j}{t}$'s. 
Hence to  obtain a more concrete bound we have to estimate 
\Q{C}{k}{i\rightarrow j}{t}
which is done in \lemref{soscontribution}
and upper bounded in \lemref{q}. 
The contributions are expressed based on a sequence of matrices $Q(t)$ defined 
below, whose role in error propagation is similar to that of the diffusion matrix in FOS.
\begin{equation}\label{eq:Qrecursion}
Q(t) = \begin{cases}
  \mathbf{I}  & \text{if } t=0 \\
  \beta \cdot M  & \text{if } t=1\\
  \beta \cdot M\, Q(t-1) + (1-\beta)\cdot Q(t-2)    & \text{if }  t\ge 2
\end{cases}
\end{equation}

\begin{lemma}\label{lem:soscontribution}
 For $t>0$, we have \[\Q{\operatorname{SOS}}{k}{i\rightarrow j}{t} = Q_{k,i}(t-1) - Q_{k,j}(t-1) \enspace . \]
\end{lemma}
\begin{proof}
Let $\bx(0)=\bx'(0) = \bi$,  $\by(0) = \mathbf{0}_{n\times n}$ and let $\by'(0)$ be also all zero except $\by'_{i,j}(0)=1$, so that $\bx(1)=\bi$, $\bx'(1) = \bj$. Let $\bx(t+1)$ and $\bx'(t+1)$ be the load vectors obtained from applying SOS for $t$ rounds on  $(\bx(1), \by(0))$ and $(\bx'(1), \by'(0))$, respectively. 
Let $\Qv{\operatorname{SOS}}{i\rightarrow j}{t}$ be a vector that has  $\Q{\operatorname{SOS}}{k}{i\rightarrow j}{t}$ as its $k$'th entry, for $1\le k\le n$. Let  $\be= \hat{\mathbf{i}}-\hat{\mathbf{j}}$.
Then we have
\begin{equation*}\label{eq:Qrecursion2}
\Qv{\operatorname{SOS}}{i\rightarrow j}{t} = \bx(t)- \bx'(t)=
\begin{cases} 
  \mathbf{0}  & \text{if } t=0 \\
   \be  & \text{if } t=1\\
  \beta \cdot M\, \Qv{\operatorname{SOS}}{i\rightarrow j}{t-1}\\ + (1-\beta)\cdot \Qv{\operatorname{SOS}}{i\rightarrow j}{t-2}    & \text{if } t\ge 2
\end{cases}
\end{equation*}
where the third equation holds because for all $t\ge 2$, both $\bx(t)$ and $\bx'(t)$ follow the same equation $\xv{}{t} = \beta \cdot M\xv{}{t-1} + (1-\beta)\cdot \xv{}{t-2}$.
Now, it can be proved by induction that 
$\Qv{\operatorname{SOS}}{i\rightarrow j}{t} = Q(t-1) \,\be$. 
Recall that all entries of $\be$ are zero except $\be_i =  1$ and $\be_j = -1$. Therefore, for $t>0$ we get 
$\Q{\operatorname{SOS}}{k}{i\rightarrow j}{t} = Q_{k,i}(t-1) - Q_{k,j}(t-1)$.
\end{proof}

The following lemma provides a bound for the second norm of $Q(t)$, which is
later used in the proofs of theorems \ref{thm:sosdet} and \ref{thm:sosrand}. 

\begin{lemma}\label{lem:q}
Let $\beta = \beta_{\small{opt}}= 2/(1+\sqrt{1-\lambda^2})$. The following statements are true.
\begin{enumerate}
\item Eigenvectors of $Q(t)$ form a basis for $\mathbb{R}^n$.
\item Let  $\gamma =  \left(\sqrt{\beta-1}\right)^{t} (t+1)$.
Then  $\gamma$ is an upper bound on the eigenvalues of $Q(t)$ except the eigenvalue corresponding to the eigenvector $(s_1,\cdots,s_n)$.
\item  $Q(t)$ has equal column sums.
\item
Define $q(t)=\sum_{1\le j\le n} Q_{i,j}(t)$ for an arbitrary $1\le i\le n$ (note that by the statement (3), this is a valid definition).
Fix a  $1\le k\le n$, and let the  vector $\ba$ be such that \break $\ba_i =  Q_{k,i}(t)- s_k/ s\cdot q(t) $. 
Then we have
$\|\ba\|_2^2 \le 2 \,  s_{\max} (\beta-1)^{t} (t+1)^2$.
\end{enumerate}
\end{lemma}

\begin{proof}
\setcounter{paragraph}{0}
\paragraph{Proof of (1)}
First we observe that the eigenvectors of $Q(t)$ are the same as the eigenvectors of $M$.
This can be proved by an induction using the recurrence  of \eq{Qrecursion} as follows.

Suppose $\bv$ is an eigenvector of $M$ with eigenvalue $\alpha$. Then $\bv$ is
also an eigenvector of $Q(t)$ and $Q(t-1)$ by the induction hypothesis. Let
$\mu_1$ and $\mu_2$ be the corresponding eigenvalues. We have
\begin{align*}
Q(t+1) \bv &=  \beta \cdot M\, Q(t) \bv + (1-\beta)\cdot Q(t-1) \bv\\
 &=  \beta \mu_1 \cdot M\, \bv + (1-\beta) \mu_2 \,  \bv\\
 &=  \beta \mu_1 \alpha \bv + (1-\beta) \mu_2   \bv\\
 &=  (\beta \mu_1 \alpha  + (1-\beta) \mu_2)\,   \bv,
\end{align*}
which shows that $\bv$ is also an eigenvector of $Q(t+1)$.
Also, note that $M = I- L S^{-1}$ where $L$ is the Laplacian matrix of the graph and $S$ is the diagonal  matrix of speeds.
The eigenvectors of $M$ are the same as those of $L S^{-1}$.
By \cite[proof of Lemma 1]{EMP02} the eigenvectors of  $L S^{-1}$ form a basis for $\mathbb{R}^n$. Therefore the eigenvectors of $M$ and the eigenvectors of $Q(t)$  form a basis for $\mathbb{R}^n$.

\paragraph{Proof of (2)}
From the induction in the proof of statement (1) one can see that corresponding to each eigenvalue $\lambda_j$ of $M$ an eigenvalue $\gamma_j(t)$ of $Q(t)$ can be obtained according to the following recursion.
\begin{equation*}
\gamma_j(t) =
\begin{cases}
  1 & \text{if } t = 0 \\
  \beta \lambda_j  & \text{if } t = 1\\
  \beta \lambda_j\cdot \gamma_j(t-1) + (1-\beta)\cdot \gamma_j(t-2)   & \text{if }  t\ge 2
\end{cases}
\end{equation*}

Solving the above recursion we get
\begin{equation*}\label{eq:gammas}
\gamma_j(t) =
\begin{cases}
  \frac{1-(\beta-1)^{t+1}}{2-\beta} & \text{if } \lambda_j = 1\\
   \left(\sqrt{\beta-1}\right)^{t} (t+1) & \text{if }  |\lambda_j| = \lambda\\
   r^t \left(\cos (\theta t) + \sin(\theta t)\cdot \frac{\lambda_j}{\sqrt{\lambda^2-\lambda_j^2}}\right)  & \text{if }  |\lambda_j| < \lambda
 \end{cases}
\end{equation*}
where $r = \sqrt{\beta-1}$, and $0 < \theta < \pi$ is such that $\sin \theta = \sqrt{\lambda^2-\lambda_j^2}/\lambda$, and $\cos \theta = \lambda_j/\lambda$.
Note that the eigenvalue corresponding to $ \lambda_j = 1$ belongs to the eigenvector $(s_1,\cdots,s_n)$.
 Hence, it suffices to prove that in \eq{gammas}  the case $ |\lambda_j| < \lambda$  does not produce eigenvalues bigger than those obtained in the case $ |\lambda_j| = \lambda$.
Note that
\begin{align}
\gamma_j(t)  &=
 r^t \left(\cos (\theta t) + \sin(\theta t)\cdot \frac{\lambda_j}{\sqrt{\lambda^2-\lambda_j^2}}\right)\notag\\
&\le \left(\sqrt{\beta-1}\right)^t \cdot \frac{\sin((t+1)\theta)}{\sin \theta}\label{eq:sin}\\
&\le \left(\sqrt{\beta-1}\right)^t\cdot (t+1),\notag
\end{align}
where in \eq{sin} we use $\sin (nx) \le n \sin x$ for  $0< x < \pi$ and $n\in \mathbb{N}$.

\paragraph{Proof of (3)} The statement follows from a simple induction using  \eq{Qrecursion}.
The case $Q(0) = \mathbf{I}$ is trivial. $Q(1) = \beta M$ also has equal column sums, since the entries in each column of  $M$ sum to one (this is necessay to guarantee load conservation).
Suppose for all $t_1\le t$, $Q(t_1)$ has equal column sums. Let us denote this value by $q(t_1)$.
\begin{align*}
\bOne_n Q(t+1)  &=  \beta \cdot\bOne_n M\, Q(t)  + (1-\beta)\cdot \bOne_n Q(t-1) \\
 &=   \beta \cdot\bOne_n Q(t)  + (1-\beta)\cdot q(t-1)\cdot  \bOne_n \\
 &=   \beta \cdot q(t) \cdot\bOne_n   + (1-\beta)\cdot q(t-1)\cdot  \bOne_n \\
 &=   \left(\beta \cdot q(t)   + (1-\beta)\cdot q(t-1)\right)\cdot  \bOne_n
\end{align*}
which shows that all column sums of $Q(t+1)$ are equal to $\beta \cdot q(t)   + (1-\beta)\cdot q(t-1)$.

\paragraph{Proof of (4)}
Let $\hat{\ba}=\hat{\mathbf{k}} -  \frac{s_k}{s}\cdot\bOne_n$. Note that $\ba =\hat{\ba}\, Q(t) $.
Let $\bv_1,\dots,\bv_n$ be the eigenvectors of $Q(t)$ with eigenvalues $\gamma_1,\dots,\gamma_n$, and $\gamma$ be defined as in the statement (2) of the lemma.
Using the fact that $M = I-LS^{-1}$, it can be proved by induction that $S^{-1}Q(t)$ is symmetric. Hence, for each right eigenvector $\bv_i$ of $Q(t)$ there is a left eigenvector $\bu_i = S^{-1}\bv_i$ with the same eigenvalue $\gamma_i$ as proved in the following.
\begin{align*}
(Q(t))^T \bu_i &=  (Q(t))^TS^{-1} S \bu_i 
  = ( S^{-1}Q(t))^T \bv_i \\
  &= S^{-1}Q(t)\bv_i
   = \gamma_i S^{-1}\bv_i  =  \gamma_i\bu_i
\end{align*}
Also, note that $S^{-1}Q(t)S\bu_i = S^{-1}Q(t)\bv_i = \gamma_i S^{-1}\bv_i = \gamma_i \bu_i $.
As a result, $\bu_i$'s are eigenvectors of $S^{-1}Q(t)S$, which is symmetric because it is the product of  symmetric matrices $S^{-1}Q(t)$ and $S$.
Therefore, $\bu_1,\dots,\bu_n$ form an orthonormal basis; so  we can write $\hat{\ba} = \sum_{i=1}^{n} c_i \bu_i$.
Now we write
\begin{equation*}\ba = \hat{\ba}\,Q(t) = \sum_{i=1}^{n}  c_i \bu_i  Q(t)  =\sum_{i=1}^{n} \gamma_i c_i \bu_i \enspace .\end{equation*}
Therefore,
\begin{equation}
\|\ba\|_2^2 = \sum_{i=1}^{n}\gamma_i^2 c_i^2 \| \bu_i\|_2^2 \le \gamma^2\sum_{i=1}^{n} c_i^2 \| \bu_i\|_2^2 = \gamma^2 \| \hat{\ba}\|_2^2 \label{eq:hatba}
\end{equation}
where the inequality uses the fact that $\hat{\ba} \perp (s_1,\dots,s_n)$
 and part (2) of the lemma, and the last equality follows from the fact that $\bu_i$'s form an orthonormal basis. Also,
\begin{align*}
 \| \hat{\ba}\|_2^2
&\le n\cdot  \frac{s_k^2}{s^2}+1 \\
&\le \frac{n \,s_k^2}{(n-1+s_k)^2}+1  \\
&\le   \frac{n\, s_k^2}{2 (n-1) s_k }+1 \\
&\le   s_k+1 \\
&\le  2 s_{\max} \enspace .
\end{align*}
Together with \eq{hatba}, this yields
$\|\ba\|_2^2\le 2 \,s_{\max}(\beta-1)^{t} (t+1)^2$.
\end{proof}


 \subsection{Deviation between Continuous and Discrete SOS  }\label{sec:sosdet}
In this section we show a bound on the deviation between a continuous SOS and 
its rounded version. 
The authors of \cite{EMS06} show a similar bound 
on the deviation using the second norm, i.e., they show  a bound of 
$||\x{D(\operatorname{SOS})}{}{t} - \x{\operatorname{SOS}}{}{t}||_2 = O\left(d \sqrt{n s_{\max}}/(1-\lambda)\right)$. Note that the bound on the deviation for FOS, which is
$O(d \sqrt{s_{\max}\log n  /(1-\lambda)})$,
is smaller.

\begin{theorem}\label{thm:sosdet}
 Consider a discrete SOS process $D = {D(\operatorname{SOS})}$ with optimal $\beta$ and a rounding scheme that rounds a fractional value to either its floor or its ceiling. Then for arbitrary $t\ge 0$ we have
$\left|\x{D}{k}{t} - \x{\operatorname{SOS}}{k}{t}\right| = 
O\left(d \sqrt{n s_{\max}}/(1-\lambda)\right)$.
\end{theorem}

\begin{proof}
We use \lemref{deterministic} to obtain a deviation bound for the general case.
We have
%
%
\begin{align}
\MoveEqLeft \left|\x{D(\operatorname{SOS})}{k}{t+1} - \x{\operatorname{SOS}}{k}{t+1}\right| \notag\\
&=  \left|\sum_{s=0}^t \sum_{\{i,j\}\in E} \left(Q_{k,i}(s) - Q_{k,j}(s)\right) \cdot \e{i}{j}{t-s} \right|\notag\\
&\le  \sum_{s=0}^t  \sum_{\{i,j\}\in E} \left|Q_{k,i}(s)-Q_{k,j}(s)\right|\notag\\
&\le  \sum_{s=0}^t  \sum_{\{i,j\}\in E} \left(\left|Q_{k,i}(s)-\frac {s_k} s\cdot q(s)\right| \right. \notag \\
& \phantom{{}={}} + \left. \left|Q_{k,j}(s)-\frac {s_k} s\cdot q(s)\right|\right)\notag \\
&=   d\cdot \sum_{s=0}^t  \sum_{i=1}^n \left|Q_{k,i}(s)-\frac {s_k} s\cdot q(s)\right|\notag\\
&\le   d\cdot \sqrt{n}\cdot \sum_{s=0}^t \left( \sum_{i=1}^n \left(Q_{k,i}(s)-\frac {s_k} s\cdot q(s)\right)^2\right)^{1/2}\label{eq:cauchy}\\
&\le   4\cdot d\cdot \sqrt{n}\cdot \sqrt{2 s_{\max}} \cdot\sum_{s=0}^\infty\left(\sqrt{\beta-1}\right)^{s} (s+1)   \label{eq:basedonq}\\
&\le   4\cdot d\cdot \sqrt{2n s_{\max}} \cdot \frac{1}{\left(1- \sqrt{\beta-1}\right)^2}   \notag\\
&\le  16\cdot d\cdot \sqrt{2n s_{\max}} \cdot \frac{1}{1-\lambda}   \notag
\end{align}
where  Equation \ref{eq:cauchy} follows from the Cauchy-Schwarz inequality and Equation \ref{eq:basedonq} follows from \lemref{q}.(4).

\end{proof}

\subsection{Framework for Randomized SOS}\label{sec:sosrand}

In the next theorem we bound the deviation between continuous and discrete 
SOS  using the randomized rounding scheme from \secref{rand}.
As  mentioned earlier in this section, is easy to see that the 
proof of Theorem  \ref{thm:randomized} holds for the more general definitions 
of linearity and contribution 
of this section. Hence, we can state the following observation and show similar to \secref{general} the next theorem.

\begin{observation}\label{obs:randomized2}
In the setting of Section \ref{sec:sos} for an arbitrary 
round $t$ we have w.h.p.
\[
\left|\X{R}{k}{t} - \x{{C}}{k}{t}\right| = O\left(  \Upsilon^{C}(G) \, \cdot \sqrt{d\log n} \,\right)
\]
\end{observation}

Similar to \secref{general}, we can use \autoref{obs:randomized2} to show the
next theorem.

\begin{theorem}\label{thm:sosrand}
Let $R = R(\operatorname{SOS})$ be a randomized-rounding discrete SOS process 
with optimal $\beta$ obtained using our randomized rounding scheme. Then 
\begin{enumerate}
\item[(1)]  $ \Upsilon^{\operatorname{SOS}}(G) =  O\left({ \sqrt{d}\cdot \log s_{\max}}/{\left(1-\lambda\right)^{3/4} }\right)\enspace $.
\item[(2)] The deviation of $R$  from the continuous SOS  in an arbitrary round $t$ is w.h.p. 
\[\left|\x{R}{k}{t} - \x{\operatorname{SOS}}{k}{t}\right| = O \Big(  \frac{ d\cdot \log s_{\max}\cdot\sqrt{\log n}}{(1-\lambda)^{3/4} } \Big) \enspace . \]
\end{enumerate}
\end{theorem}

\begin{proof}
The bound on the refined local divergence is obtained using the formulation of
\lemref{soscontribution} and the bound in \lemref{q}. This bound together with
the parametric deviation bound of \thmref{randomized} yield the second
statement of the theorem.

We write
\begin{align}
 \left(\Upsilon^{\operatorname{SOS}}(G)\right)^2
&= \sum_{t=0}^{\infty} \sum_{i=1}^{n} \max_{j \in N(i)} \big ( Q_{k,i}(t) - Q_{k,j}(t) \big)^2 \notag\\
&\le  \sum_{t=0}^{\infty}\sum_{i=1}^n \sum_{j\in N(i)}   \left( Q_{k,i}(t) - Q_{k,j}(t) \right)^2 \notag\\
&=  \sum_{t=0}^{t_1-1} \sum_{i=1}^n \sum_{j\in N(i)}   \left( Q_{k,i}(t) - Q_{k,j}(t) \right)^2 \notag\\
& \phantom{{}={}}+ \sum_{t=t_1}^{\infty}\sum_{i=1}^n \sum_{j\in N(i)}
\left( Q_{k,i}(t) - Q_{k,j}(t) \right)^2\notag
\end{align}

In the following, we use $\sigma=\beta-1$ for brevity. Note that we have $0<\sigma<1$.
{
\allowdisplaybreaks
\begin{align}
  \MoveEqLeft \sum_{t=0}^{t_1-1} \sum_{i=1}^n \sum_{j\in N(i)}
  \left( Q_{k,i}(t) - Q_{k,j}(t) \right)^2\notag\\
 &\le  \sum_{t=0}^{t_1} \sum_{i=1}^n \sum_{j\in N(i)} 2 \left(  (Q_{k,i}(t))^2+ (Q_{k,j}(t) )^2  \right)  \notag\\
 &\le  4 \cdot d\cdot \sum_{t=0}^{t_1-1} \sum_{i=1}^n(Q_{k,i}(t))^2
 \le  4 \cdot d\cdot \sum_{t=0}^{t_1-1} \|Q(t) \,\hat{\textbf{k}} \|_2^2 \notag\\
 &=  4 \cdot d\cdot \sum_{t=0}^{t_1-1}\left( \|\hat{\textbf{k}} \|_2\cdot \max_i \gamma_i(t)\right)^2\notag\\
 &=  4 \cdot d\cdot \sum_{t=0}^{t_1-1}\left(\frac{1-\sigma^{t+1}}{1-\sigma}\right)^2\notag\\
 &\le  4 \cdot d\cdot t_1\cdot \left(\frac{1-\sigma^{t_1}}{1-\sigma}\right)^2
 \le  4 \cdot d\cdot t_1\cdot \left({1-\sigma}\right)^{-2}\label{eq:sosupshead}
\end{align}
}
We also get
\begin{align}
\MoveEqLeft \sum_{t=t_1 }^{\infty}\sum_{i=1}^n \sum_{j\in N(i)}   \left( Q_{k,i}(t) - Q_{k,j}(t) \right)^2 \notag\\
&\le   \sum_{t=t_1 }^\infty  \sum_{i=1}^n \sum_{j\in N(i)} 2\left(Q_{k,i}(t) - \frac {s_k} s\cdot q(t)\right)^2\notag \\
& \phantom{{}={}}  + 2 \left(Q_{k,j}(t)- \frac {s_k} s\cdot q(t)\right)^2\notag\\
&=  4\cdot  d\cdot \sum_{t=t_1 }^\infty  \sum_{i=1}^n \left(Q_{k,i}(t)- \frac {s_k} s\cdot q(t)\right)^2\\
&\le  8\cdot  d\cdot s_{\max}\cdot \sum_{t=t_1 }^\infty \left(\sigma^{t} \,(t+1)^2\right)  \label{eq:upsilonmiddle}
\end{align}
where the last inequality follows from part (4) of \lemref{q}. The above
summation can be bounded as follows.
\begin{align}
\MoveEqLeft \sum_{t=t_1}^\infty \left(\sigma^{t}\, (t+1)^2\right) 
= \frac{d}{d\sigma} \left( \sigma \frac{d}{d\sigma} \left(\frac{\sigma^{t_1 + 1}}{1-\sigma}\right) \right) \notag\\
&\le \frac{2{\sigma}^{t_1+2}}{{\left(1-\sigma\right)}^{3}}+\dfrac{\left(2t_1+3\right){\cdot}{\sigma}^{t_1+1}}{{\left(1-\sigma\right)}^{2}}
+\frac{{\left(t_1+1\right)}^{2}{\cdot}{\sigma}^{t_1}}{1-\sigma}
 \label{eq:sosupstail}
\end{align}

Let $t_1 = (\log s_{\max})/(1-\sigma) $. Note that $\sigma^{1/(1-\sigma)} \le 1/e$. Then \eq{sosupstail} yields
\begin{align}
 \sum_{t=t_1}^\infty \left(\sigma^{t}\, (t+1)^2\right)
&=O\left(\frac{\log^2 s_{\max}}{s_{\max}\cdot(1-\sigma)^3}\right) \enspace .
 \label{eq:sosupstail2}
\end{align}

Combining equations \eqref{eq:sosupstail2} and \eqref{eq:upsilonmiddle} and then \eqref{eq:sosupshead} we get
\begin{align}
 \left(\Upsilon^{\operatorname{SOS}}(G)\right)^2
&= O\left(\frac{d\cdot \log s_{\max}}{\left(1-\sigma\right)^{3} }\right)  + O\left(\frac{d\cdot \log^2 s_{\max}}{\left(1-\sigma\right)^{3} }\right) \notag\\
&= O\left(\frac{d\cdot \log^2 s_{\max}}{\left(1-\sigma\right)^{3} }\right) \enspace .\notag
\end{align}
Observe that
\begin{equation*} 1-\sigma = (1-\sqrt{\sigma}) (1+\sqrt{\sigma}) \ge  (1-\sqrt{\sigma})=  1-\sqrt{\beta-1} \enspace , \end{equation*}
where we have
\begin{align*}
 1-\sqrt{\beta-1}
&=  1-\frac{\lambda}{1+\sqrt{1-\lambda^2}}\\
&=  \frac{1-\lambda+\sqrt{1-\lambda^2}}{1+\sqrt{1-\lambda^2}} \\
&\ge \frac{1}{2}\cdot \left({1-\lambda+\sqrt{1-\lambda^2}}\right) \\
&\ge \frac{1}{2}\cdot {\sqrt{1-\lambda}\cdot\left(\sqrt{1-\lambda}+\sqrt{1+\lambda}\right)} \\
&\ge \frac{1}{2}\cdot {\sqrt{1-\lambda}} \enspace .
\end{align*}
Therefore,
\begin{equation*}
\Upsilon^{\operatorname{SOS}}(G) =O\left(\frac{ \sqrt{d}\cdot \log s_{\max}}{\left(1-\lambda\right)^{3/4} }\right) \enspace .
\end{equation*}
This finishes the proof of the first statement. The bound in the second statement follows immediately from statement (1) and  \thmref{randomized}.
\end{proof}



%% file: negative.tex
\section{Negative Load for SOS}\label{sec:negload}
In second order diffusion nodes might not have enough load to satisfy all
their neighbors' demand. 
This situation, which we refer to as \emph{negative load}, motivates studying by 
how much a node's load may become negative. 
Here we study the  
minimum amount of  load that nodes need in order  to prevent this event.
In the following we calculate a bound on the minimum load of every node that 
holds during the whole balancing process. Note that, if every processor has 
such a minimum load at the beginning of the balancing process, there will be no 
processor with negative load. Hence, these bounds can also be regarded as 
bounds on the minimum load of every processor in order to avoid negative load.

Let $\bar{x} = (\bar{x}_1,\dots,\bar{x}_n)$ be the balanced load vector. Define $\Delta(t) = \| x(t)-\bar{x}\|_\infty$, and $\Phi(t)= \|x(t)-\bar{x}\|_2$, where $\|.\|$ is the norm operator.
Then the following observation estimates the load at the end of every step.

\begin{observation}\label{obs:soslb}
In continuous SOS  with $\beta = \beta_{\small{opt}}$ we have \begin{equation*} x(t) \ge -\sqrt{n} \cdot\Delta(0) \enspace .\end{equation*}
\end{observation}

\begin{proof}
We first note that
\begin{equation}\label{eq:deltat}
 \Delta(t) \le \Phi(t) 
\stackrel{(*)}{\le} \lambda^{t} \cdot \Phi(0)
\le  \lambda^{t} \cdot\sqrt{n}\cdot \Delta(0) \enspace ,
\end{equation}
where $(*)$ follows from a result by Mutukrishnan et al.~\cite{MGS98}. They show that $x(t) = M(t)\, x(0)$ for an $n\times n$ matrix $M(t)$ defined recursively. 
They also show that $\Phi(t) \le \gamma(t) \cdot \Phi(0) $  where $\gamma(t)$ is the second largest eigenvalue in magnitude of $M(t)$ and
$\gamma(t) \le  \lambda^{t}$ \cite[Proof of Theorem 2]{MGS98}. 
Though they only consider homogeneous networks, their  argument also applies to the heterogeneous case.
The proof now follows by considering the facts $-x(t) \le \Delta(t) $ and $\lambda^{t} < 1$.
\end{proof}

It should be noted that the load during a single balancing step can be lower than the bound given in
in \obsref{soslb} since \obsref{soslb} 
considers only snapshots of the network at the end of each round. 
It might be possible that a node has to send more load items to some of its
neighbors than it has at the beginning of round $t$, but still its load remains
positive at the end of round $t$. This can happen if it also receives many load 
items from other neighbors in round $t$. To study the negative load issue it 
is helpful to divide every round in  two distinct steps, where in the first step
 all nodes send out their outgoing flows. In the 
second step, they receive incoming flows sent by their neighbors in the
first step. At the end of the first step  
all the outgoing flows are sent out but no incoming flow is yet received. 
To prevent negative load the load of every node has to be non-negative at  
this point. We call this state the \emph{transient state} and use $\breve{x}_i(t)$ to denote 
the load  in the transient state.
Note that we always have $\breve{x}_i(t) \le \x{}{i}{t}$ and $\breve{x}_i(t) \le \x{}{i}{t+1}$.
The following theorem provides a lower bound on $\breve{x}_i(t)$.

\begin{theorem}\label{thm:negload}
In a continuous SOS process with $\beta = \beta_{\small{opt}}$ we have
  $\breve{x}_i(t) \ge -O\left( {\sqrt{n}\cdot \Delta(0) }/{\sqrt{1-\lambda}}\right)$.
\end{theorem}
\begin{proof}
Observe that for $t>1$ and an arbitrary node $i$
\begin{equation}
 \y{}{i}{j}{t} = (\beta-1)\cdot \y{}{i}{j}{t-1}+ \beta\cdot \alpha_{i,j}\cdot \left(\frac{\x{}{i}{t}}{s_i}-\frac{\x{}{j}{t}}{s_j}\right) \label{sosyrec}
\end{equation}
and since $\bar{x_i}/s_i = \bar{x}_j/s_j$ we have
\begin{align*} 
\MoveEqLeft \sum_{j\in N(i)}|\y{}{i}{j}{t}|\\ 
&\le (\beta-1)\cdot\sum_{j\in N(i)}|\y{}{i}{j}{t-1}| \\
&\phantom{{}={}} + \beta\cdot \sum_{j\in N(i)}\alpha_{i,j}\cdot \left|\frac{\x{}{i}{t}}{s_i}-\frac{\x{}{j}{t}}{s_j}\right|\\
&\le  (\beta-1)\cdot\sum_{j\in N(i)}|\y{}{i}{j}{t-1}| \\
&\phantom{{}={}} + \beta\cdot \sum_{j\in N(i)}\alpha_{i,j}\cdot \left(\left|\frac{\x{}{i}{t}}{s_i}-\frac{\bar{x}_i}{s_i}\right| + \left|\frac{\x{}{j}{t}}{s_j}-\frac{\bar{x}_j}{s_j}\right|\right) \enspace .
\end{align*}

Let  $g(t)=\sum_{j\in N(i)}|\y{}{i}{j}{t}|$. Recall that $\beta<2$, and for all $i$, $s_i\ge 1$  and $\sum_{j\in N(i)\cup\{i\}}\alpha_{i,j} = 1$. So we get
\begin{align}
\MoveEqLeft g(t+1)  \notag \\
&\le (\beta-1)\cdot g(t) \notag \\
& \phantom{{}={}} + 2\sum_{j\in N(i)}\alpha_{i,j} \left(\left|\x{}{i}{t+1}-\bar{x}_i\right| + \left|\x{}{j}{t+1}-\bar{x}_j\right|\right)\notag\\
&\le  (\beta-1)\cdot g(t) + 4\cdot \Delta(t+1)\cdot\sum_{j\in N(i)}\alpha_{i,j} \notag\\
&\le  (\beta-1)\cdot g(t) + 4\cdot \Delta(t+1) \notag\\
\intertext{and by \eq{deltat}} 
&\le  (\beta-1)\cdot g(t) + 4\cdot  \lambda^{t+1} \cdot\sqrt{n}\cdot \Delta(0) \enspace . \label{eq:grec}
\end{align}
Note that $g(0) <\Delta(0)$. From the recurrence of \eq{grec} we obtain
\begin{align}
g(t+1) &\le 4\sum_{i=0}^t (\beta-1)^{t-i} \cdot \lambda^i \cdot \sqrt{n}\cdot \Delta(0)\notag\\
&=  4\,\sqrt{n}\cdot \Delta(0) \cdot \frac{\lambda^{t+1} - (\beta-1)^{t+1}}{\lambda - (\beta-1)}\notag\\
&\le 4\,\sqrt{n}\cdot \Delta(0) \cdot \frac{\lambda}{\lambda - (\beta-1)}\label{eq:gexpansion}
\end{align}
where the last inequality holds because $\lambda <1$.
On the other hand, we have
\begin{align*}
\lambda - (\beta-1) &= \frac{(1+\lambda) \sqrt{1-\lambda^2}-(1-\lambda)}{1+\sqrt{1-\lambda^2}} \\
&> \frac{\sqrt{1-\lambda}\cdot (\sqrt{(1+\lambda)^3}-\sqrt{1-\lambda})}{2}\\
&> \frac{\sqrt{1-\lambda}\cdot (1-\sqrt{1-\lambda})}{2}\\&
> \frac{\sqrt{1-\lambda}\cdot \lambda}{4} \enspace .
\end{align*}
Therefore, we can apply the above to  \eq{gexpansion} to get the bound
$g(t) =O\left( {\sqrt{n}\cdot \Delta(0) }/{\sqrt{1-\lambda}}\right)$.

To complete the proof, we note that
 $\breve{x}_i(t) \ge \x{}{i}{t}-g(t),$
 while by \obsref{soslb} we have $ \x{}{i}{t} \ge -\sqrt{n} \cdot\Delta(0)$. This yields the
lower bound of $-O\left( {\sqrt{n}\cdot \Delta(0) }/{\sqrt{1-\lambda}}\right)$.
\end{proof}

The next result shows that the asymptotic lower bound  obtained in \obsref{soslb} also holds for the randomized discrete second-order process R = R(SOS) in many cases, for instance, when $s_{\max}$ is polynomial in $n$ and $d/(1-\lambda)^{3/4} = O(n^{0.5-\varepsilon})$ for some $\varepsilon>0$.
This is true, e.g., for tori with four or more dimensions, hypercubes, and expanders.
Then we can apply a similar argument as in the proof of \thmref{negload} to get a lower bound for $R$.

\begin{theorem}\label{thm:negloaddisc}
In a discrete SOS process $R=R(\operatorname{SOS})$ with $\beta = \beta_{\small{opt}}$,
$s_{\max}$ polynomial in $n$, and $d/(1-\lambda)^{3/4} = O(n^{0.5-\epsilon})$ for some 
$\epsilon > 0$, we have
  \[\breve{x}_i^{R}(t) \ge -O\left( \frac{\sqrt{n}\cdot \Delta(0) +d^2}{\sqrt{1-\lambda}}\right) \enspace .\]
\end{theorem}

\begin{proof}
To show this result we first rewrite \eq{sosy} as follows.
%
%
\begin{align*}
 \y{}{i}{j}{t} & \le
 (\beta-1)\cdot \y{}{i}{j}{t-1} \\
& \phantom{{}={}} +\beta\cdot \alpha_{i,j}\cdot \left(\frac{\x{}{i}{t}}{s_i}-\frac{\x{}{j}{t}}{s_j}\right) +d \enspace ,
\end{align*}
resulting in
\begin{align}
g(t+1) &\le (\beta-1)\cdot g(t) + 4\cdot  \lambda^{t+1} \cdot\sqrt{n}\cdot \Delta(0) +d^2 \enspace .\notag
\end{align}
Then we rewrite \eqref{sosyrec} in the proof of \thmref{negload} as follows.
\begin{align}
 \y{}{i}{j}{t}
&\le (\beta\!-\!1)\cdot \y{}{i}{j}{t\!-\!1}+ \beta\cdot \alpha_{i,j}\cdot \left(\frac{\x{}{i}{t}}{s_i}\!-\!\frac{\x{}{j}{t}}{s_j}\right) +d \notag
\end{align}
which gives us
\begin{align}
g(t+1) &\le (\beta-1)\cdot g(t) + 4\cdot  \lambda^{t+1} \cdot\sqrt{n}\cdot \Delta(0) +d^2 \enspace .\notag
\end{align}
Proceeding with similar steps as in the proof of \thmref{negload} we get the
 following bound for D. 
\[\breve{x}_i^{D}(t) \ge -O\left( \frac{\sqrt{n}\cdot \Delta(0) +d^2}{\sqrt{1-\lambda}}\right) \enspace . \]

\end{proof}

%% file: simulation.tex
\FloatBarrier
\section{Simulation Results} \label{sect:simulation}

In this section we present some simulation results for several balancing
algorithms. We simulated discrete versions of both, first order and second
order balancing schemes where we use \emph{randomized rounding} as described in
Section \ref{sec:rand} for the discretization. 
Our main goal is to see under which circumstances SOS outperforms FOS.

We consider different networks which are based on various graph classes.
A complete list of all graph types and parameters used for the
simulation can be obtained from Table \ref{tab:graph-types}. 
%
Our simulation tool is highly modularized and
supports various load balancing schemes and rounding procedures. It can be
used to simulate the load balancing process using multiple threads on a
shared-memory machine. To fully utilize the capability of modern CPUs we used
OpenMP to generate code that performs suitable instructions in parallel.
The simulation was implemented using the
C\kern-.02em\raise.667ex\hbox{\footnotesize{++}} programming language. Our
tests were conducted on an Intel Core i7 machine with 4 cores and 8 GB
system memory. 

If not stated otherwise we initialize our system by assigning a load of $1000
\cdot n$ to a fixed node $v_0$, where $n$ is the number of nodes of the network,
and the load of all other nodes is set to zero. Our data plotted in \autoref{fig:2d-torus-1000x1000-initial-loads}, however, indicate
that the amount of initial load does only have limited impact on the behavior
of the simulation, especially once the system has converged.

We investigate the following metrics measuring the quality of the load
distribution.
\begin{enumerate}[leftmargin=*]
\item \textbf{Maximal local load difference.} This is
the maximum load difference between the nodes connected by an edge.
 That is, the maximum local load difference for given
load vectors $x(t)$ in a round $t$ is defined as
\begin{equation*}
\phi_{\text{local}}\left(x(t)\right) = \max_{\left\{u, v\right\} \in E}\left\{ \left| x_u(t) - x_v(t) \right| \right\} \enspace .
\end{equation*}

\item \textbf{Maximum load.}  
This is the maximum load of any node minus the average load $\overline{x}$.
\begin{equation*}
\phi_{\text{global}}\left(x(t)\right) = \Delta(t) = \max_{v \in V}\left\{ x_v(t) \right\} - \overline{x}
\end{equation*}

\item \textbf{Potential based on $2$-norm.} We compute the value of the potential function
$\phi_t$ proposed by Muthukrishnan et al.\ \cite{MGS98} which is defined as
\begin{equation*}
\phi_t = \phi(x(t)) = \sum_{v \in V}\left( x_v(t) - \overline{x} \right)^2
\end{equation*}
In our plots we divided this potential by $n$.

\item \textbf{Impact of eigenvectors on load.} We initially compute the
eigenvectors of the diffusion matrix and solve in each round $t$ the the linear
system $V\cdot a = x(t)$, which is defined over the orthonormal matrix of $n$
eigenvectors $V$ and the load vector $x(t)$. We then identify the
\emph{leading} eigenvector, i.e., the eigenvector with the largest $|a_i|$. The
coefficients $a_i$ for $i = 2, \dots, n$ describe together with the eigenvectors
the load imbalance completely \cite{GV61}. Observe that the coefficient $a_i$
in round $t$ multiplied with the corresponding eigenvalue $\mu_i$ yields the
coefficient in the following round $t+1$. 
Therefore, the largest coefficient governs the convergence rate in that step.

\item \textbf{Remaining imbalance.}
This is the remaining imbalance of the converged system (see \cite{ES10}),
i.e., the number of tokens above average once this number starts to fluctuate and does not visibly improve any
more. This imbalance does not occur in continuous systems and is due to the
applied rounding in discrete systems.
\end{enumerate}

\noindent In this first section we focus on the torus. For results w.r.t.\ other graph classes see
\autoref{sect:other}.

\subsection{Results for the Torus}

\begin{table} \centering
\caption{graph types and parameters used in simulation}
\label{tab:graph-types}
\begin{tabular}{lll}
Graph & Size & Parameter $\beta$ \\
\hline
Two-Dimensional Torus& $n = 1000\times1000$ & $1.9920836447$ \\
Two-Dimensional Torus& $n = 100\times100$ & $1.9235874877$ \\
Random Graph (CM) & $n = 10^6$, $d = \left\lfloor\log_2{n}\right\rfloor$ & $1.0651965147$\\
Random Geometric Graph & $n = 10^4$, $r = \sqrt[4]{\log{n}}$ & $1.9554636334$\\
Hypercube & $n = 2^{20}$ & $1.4026054847$\\
\end{tabular}

\end{table}

\begin{figure}
\centering
\includegraphics{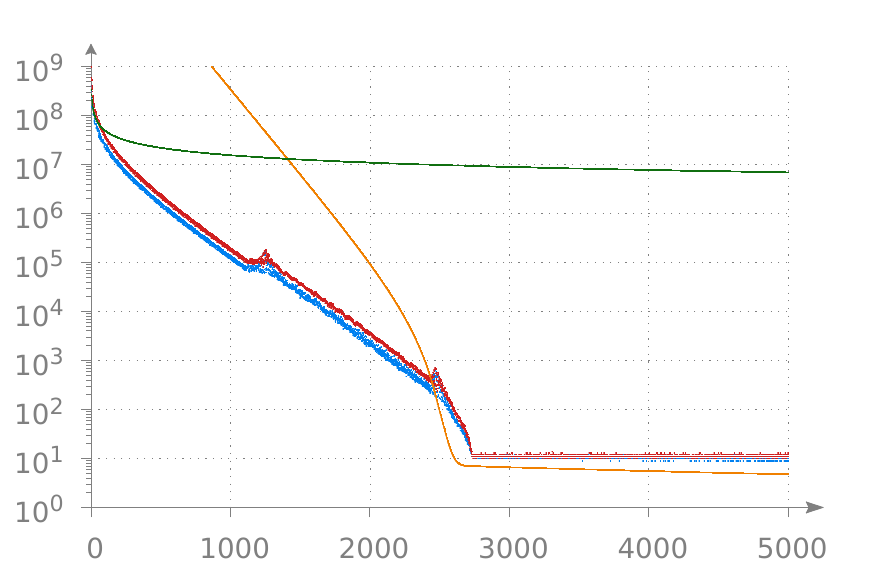}
\caption{The maximum load minus the average load is plotted in blue, the
maximum local load difference in red and the potential function $\phi_t$ in
yellow, using SOS on a two-dimensional torus size $1000 \times 1000$. As a
comparison, the green line shows the maximum load minus the average load using
FOS.}
\label{fig:sos}
\end{figure}

Our main results are shown in Figure \ref{fig:sos}, where we plotted
the simulation results using the second order scheme with randomized rounding
in a two-dimensional torus consisting of $1000 \times 1000$ nodes and an average load of $1000$. As in all
following plots, the $x$-axis represents the number of rounds. The plot shows
the maximum load minus the average load, the maximum local load difference, and
the potential function $\phi_t$ on the $y$-axis. As a comparison, a simulation run
using only first order scheme is shown as well.

\begin{figure}[b]
\centering
\includegraphics{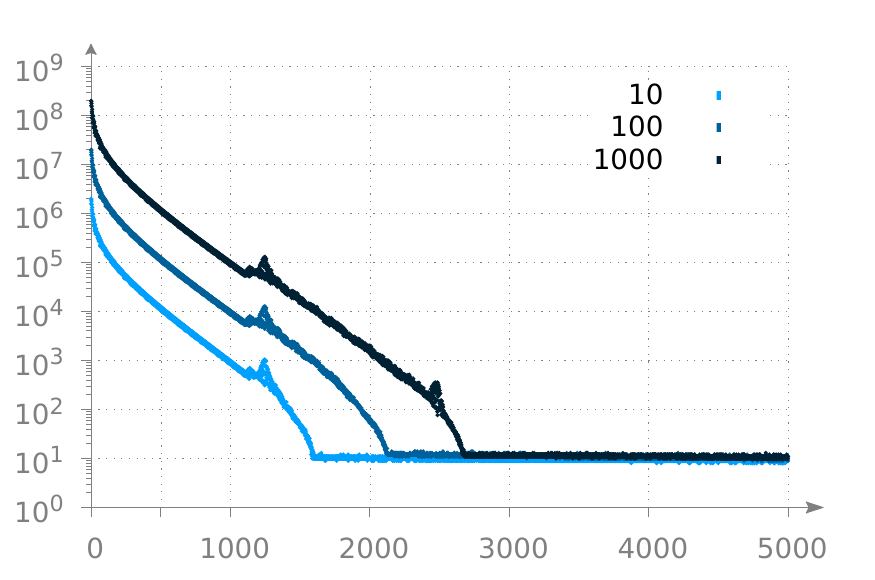}
\caption{The plot shows the maximum load minus the average load 
on a two-dimensional torus of size $1000 \times 1000$.
Three different initial loads were used with average loads of $10$, $100$, and
$1000$, colored from light to dark.}
\label{fig:2d-torus-1000x1000-initial-loads}
\end{figure}

\begin{figure}
\centering
\includegraphics[page=1]{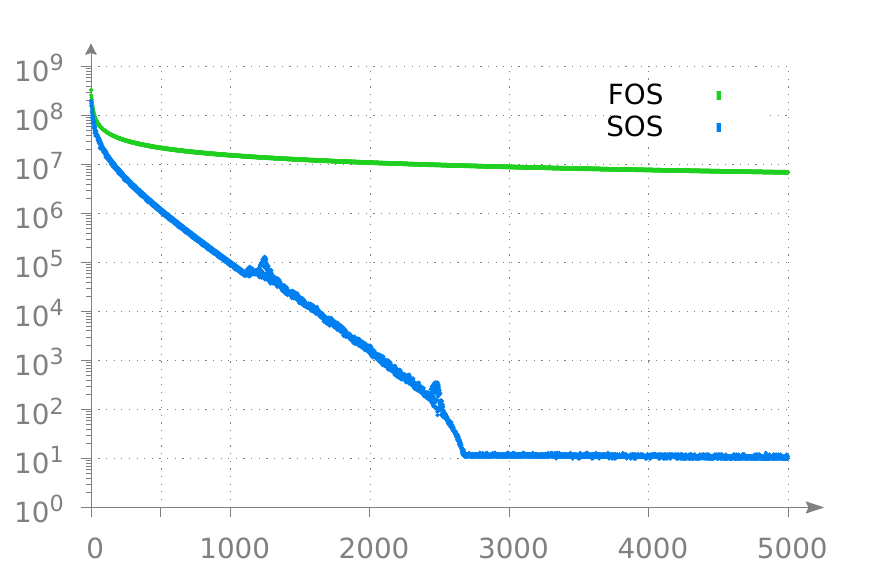}
\includegraphics[page=2]{figures/plots/2d-torus-1000x1000-fos-sos-comparison-small}
\caption{A comparison of the maximum load minus the average load using SOS
(blue) and FOS (green) on a two-dimensional torus of size $1000 \times 1000$.
The first plot shows discrete loads and randomized rounding, the second plot
shows an idealized scheme.}
\label{fig:2d-torus-1000x1000-fos-sos-comparison}
\end{figure}

\begin{figure*}
\centering
\includegraphics{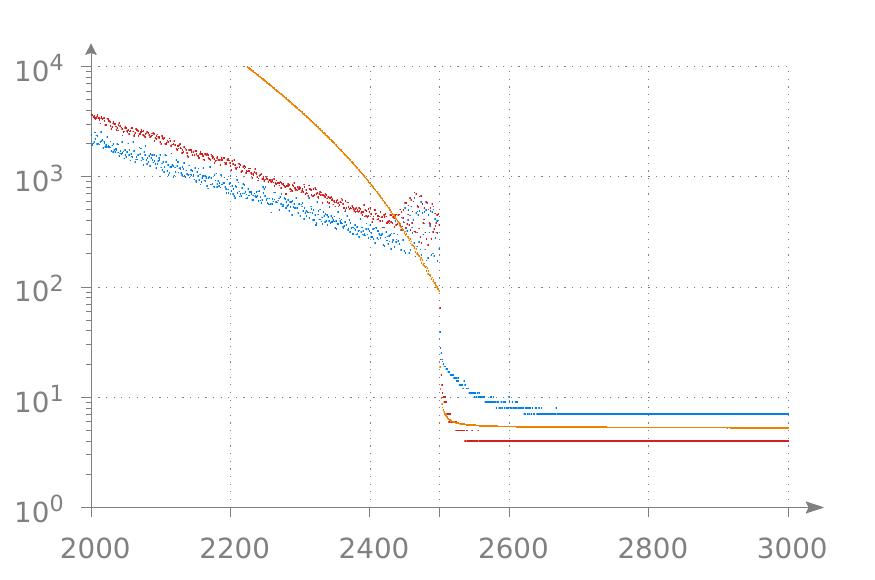} \hfill
\includegraphics{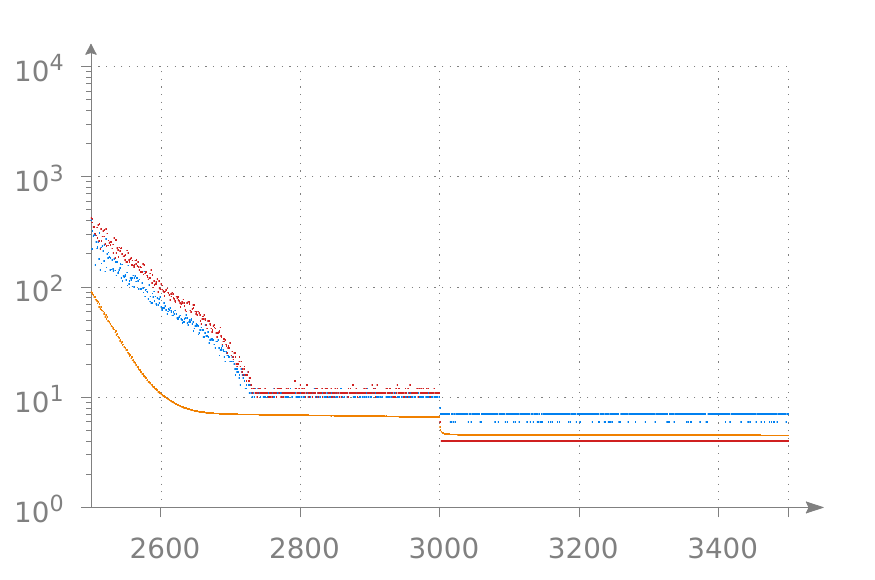}
\caption{The plots show the maximal local difference in red, the maximal load
minus the average load in blue, and the potential function $\phi_t$ in yellow.
The simulation switches from second order scheme to first order scheme in
the left and the right plot after 2500 and 3000 rounds, respectively.}
\label{fig:sos-switch-to-fos}
\end{figure*}
It is known that the second order scheme is faster than the first order scheme  w.r.t.\ the convergence time
of the load balancing system in graphs with a suitable eigenvalue gap. However, our simulations indicate that for SOS the 
remaining maximal load difference does not drop below a certain threshold.
Therefore, we implemented the following approach to decrease the load
differences even further. First we perform a number of steps using the fast second
order scheme. Then, every node synchronously switches to first
order scheme. We considered two different scenarios.  
In the first case we switched to FOS early after 2500
SOS steps. This number of steps corresponds roughly to the end of a phase of
exponential decay in the potential function. In the second case we switched to
FOS rather late at 3000 steps, allowing the system to run for a few hundred
additional steps using the second order scheme. In both cases we observed a
significant drop in both, the local and the global load differences. That is,
the values for the load differences do not drop below $10$ when using SOS. Once
the simulation is switched to FOS, the maximum local load difference converges
to a value of $4$ and the maximum load minus the average load drops to $7$.
This is shown in Figure \ref{fig:sos-switch-to-fos}; a direct comparison is
shown in Figure \ref{fig:sos-switch-to-fos-comparison}.

\begin{figure*}
\includegraphics{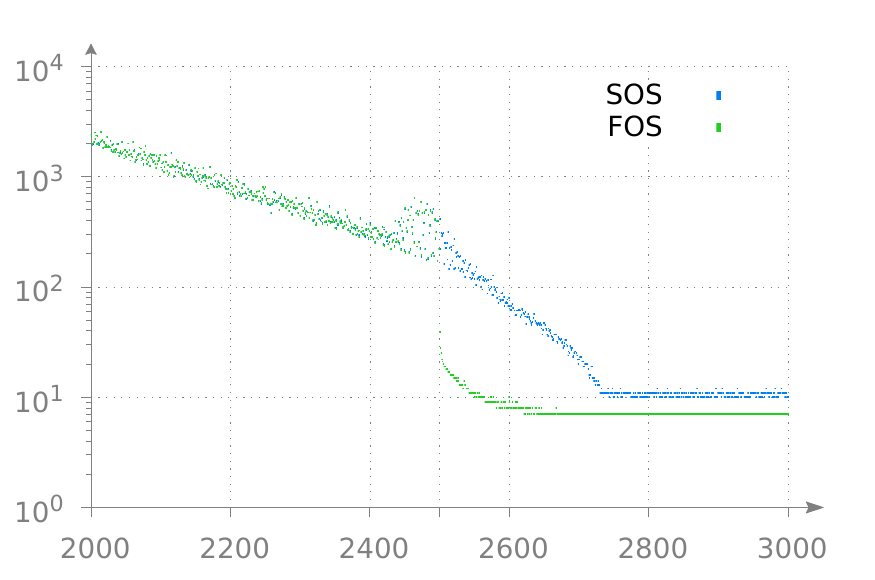} \hfill
\includegraphics{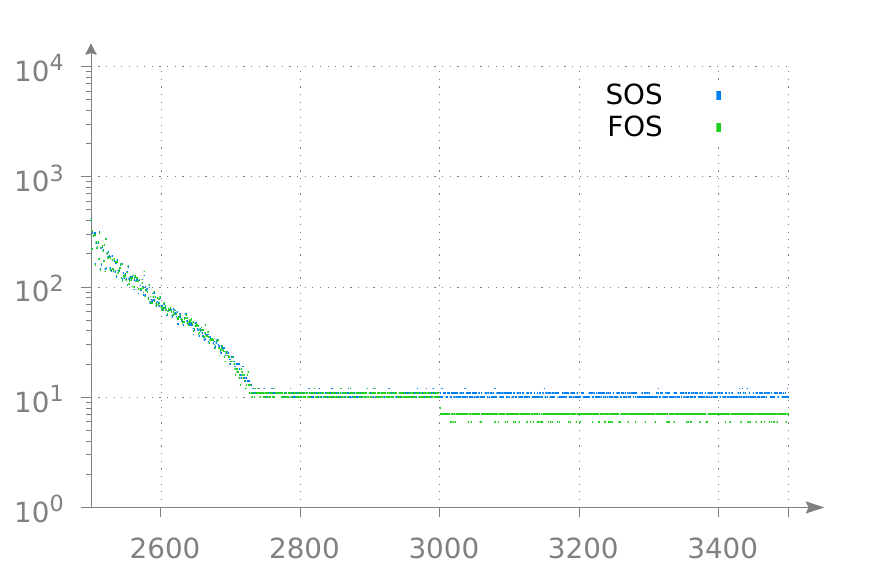}
\caption{The plots show a direct comparison of the same data presented in
Figure \ref{fig:sos-switch-to-fos}. The blue data points show the maximal load
minus the average load using only a SOS approach while the green data points
show the maximal load minus the average load when switching to FOS. Again, the
switch has been conducted after 2500 steps in the left and 3000 steps in the
right plot.}
\label{fig:sos-switch-to-fos-comparison}
\end{figure*}

\begin{figure*}
\includegraphics{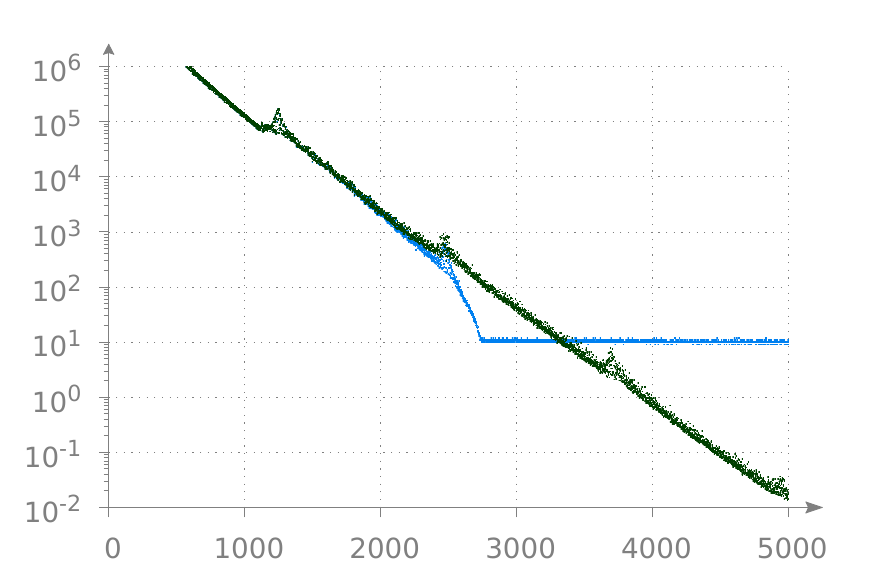} \hfill \includegraphics{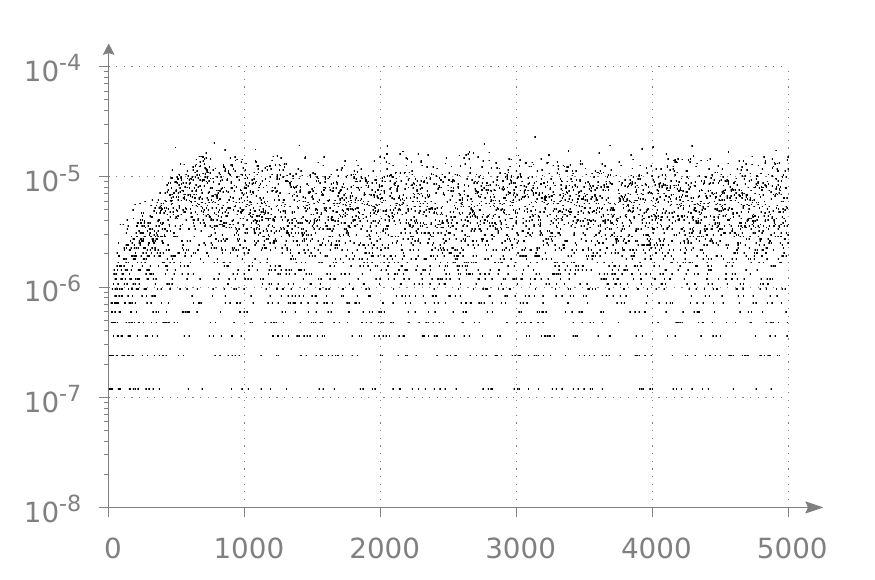}
\caption{A comparison of the idealized second order scheme in green with a SOS using
randomized rounding in blue. The idealized version is based on IEEE754 double precision
floating point values as loads. The data points show the maximum load of the
system minus the average load. The right plot shows the absolute value of the
total load in the system at round $t$ minus the initial total load, i.e., the
absolute error.}
\label{fig:comparison-idealized}
\end{figure*}

\begin{figure*}
\centering
\includegraphics{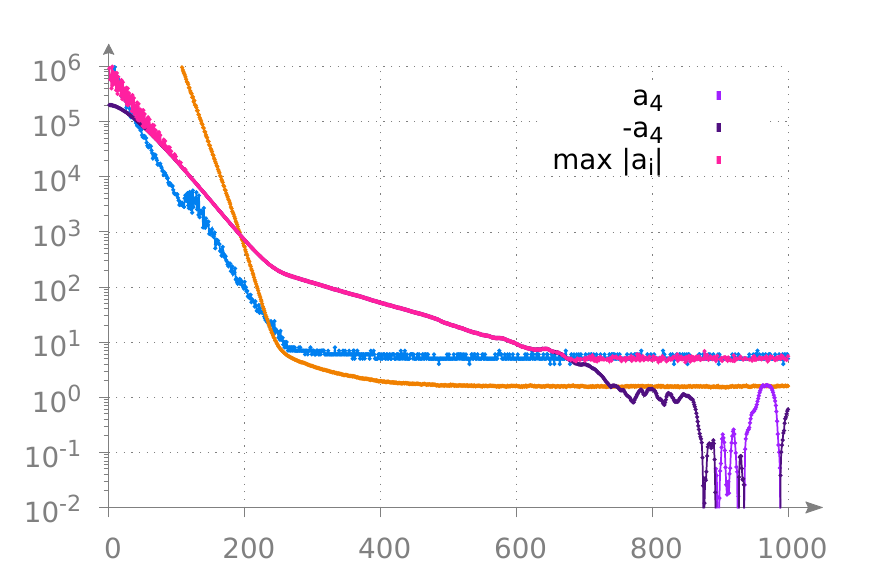} \hfill \includegraphics{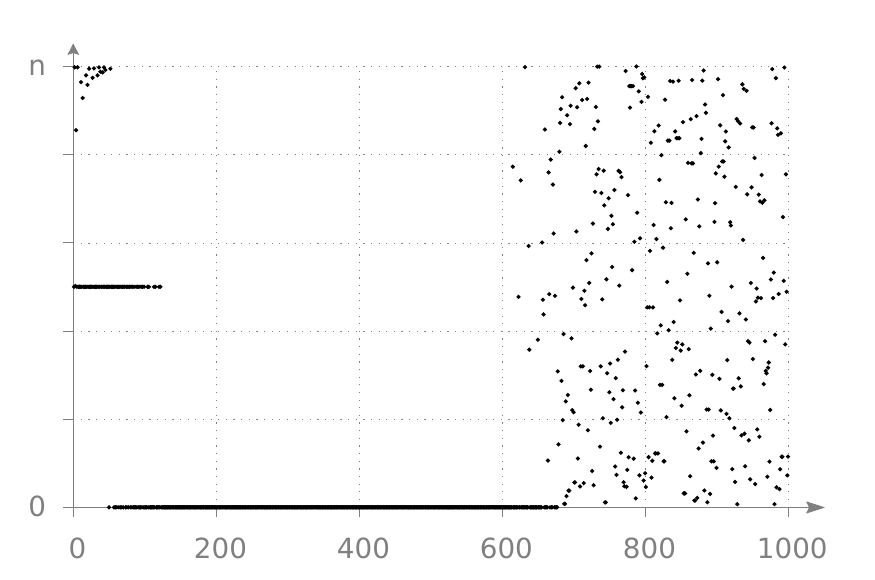}
\caption{The left plot shows the impact of eigenvectors on the load on a
two-dimensional torus of size $100\times100$. The maximum over all
coefficients, $\max_{i}\left\{|\mathfrak{a}_i|\right\}$, is shown along with
$\mathfrak{a}_4$. In the right plot the currently leading coefficient is shown,
i.e., a black point indicates that in the given round ($x$-axis) the
corresponding eigenvector ($y$-axis) has maximal impact.}
\label{fig:2d-torus-100x100-alpha}
\end{figure*}

In the left plot in Figure \ref{fig:sos-switch-to-fos} we furthermore observe that the load differences
continue to diminish for about 200 steps (during steps 2500 to 2700) when we switch to FOS after 2500 steps. When we
switch to FOS after 3000 steps (right plot in Figure \ref{fig:sos-switch-to-fos}) a drop can still be observed, however, the resulting load
differences remain at a low level. 
To explain this behavior of the load
balancing procedure we analyzed the impact of the eigenvectors of the diffusion
matrix on the load balancing process. Recall that the diffusion matrix $M = \left(M_{ij}\right)$
is defined as
\begin{equation*}
M_{ij} =
\begin{cases}
	\alpha_{ij} & \text{if } i \neq j \\
	1 - \sum_{i\neq j}\alpha_{ij} & \text{if } i = j
\end{cases}
\end{equation*}
with $\alpha_{ij} = 1 / \left(\max\left\{\deg(i), \deg(j)\right\} + 1\right)$
if node $i$ is adjacent to node $j$ and $0$ otherwise.

We used the Lapack library \cite{lapack} to compute the eigenvalues and
corresponding eigenvectors of $M$. The same library was then used to solve the
set of linear systems
\begin{equation*} V \cdot \mathfrak{a} = W \end{equation*}
for a matrix of coefficients $\mathfrak{a}$, where $V$ denotes a matrix of
eigenvectors of $M$ and $W = \left(x(t)\right)$ consists of row vectors
$x(t)$ as defined in \autoref{sec:known} containing the loads of the system at every round $t$. The resulting
coefficients in $\mathfrak{a}$ give the impact of the corresponding eigenvectors in
each round on the load. The results are shown for the torus of size
$100\times100$ in the two plots of Figure \ref{fig:2d-torus-100x100-alpha}. The
first plot shows the maximum of these coefficients. In the simulation run corresponding to this plot we observed that
starting roughly after 100 rounds this leading eigenvector corresponds to
$\mathfrak{a}_4$ up until roughly round 700. After 
that time there is no clear
leading eigenvector. This can be observed from the right plot in the same
figure, where the currently leading coefficient is plotted for each round.

It seems reasonable to switch from SOS to FOS once the impact of the leading eigenvector
drops below some threshold. This information, however, requires a global view on
the load balancing network and therefore cannot be used in a distributed
approach. In real-world applications also the trade-off between a remaining
imbalance and the time required to balance the loads must be considered. We
therefore investigate the effect of the time step when switching from SOS to
FOS. 

In Figure \ref{fig:2d-torus-100x100-fos} we plotted the maximum load minus the
average load for second order scheme and for an adaptive approach where we
switched to FOS after a number of SOS rounds. The impact of the leading
eigenvector (and the loss thereof) explains the data shown in Figure
\ref{fig:2d-torus-100x100-fos}. 
%
%
Our data indicate, that once the impact
of the leading eigenvector drops below a certain threshold in a round $R$, there is
no difference in the behavior of the system when switching to FOS in some
consecutive round $r \geq R$. Independently of the round $R$, however, we
observe a significant drop in the maximum load.

Note that the maximum local load difference seems to be a good indicator for
switching from SOS to FOS. Furthermore this local property is also available in
a distributed system with only limited global knowledge.

\begin{figure*}
\centering
\includegraphics{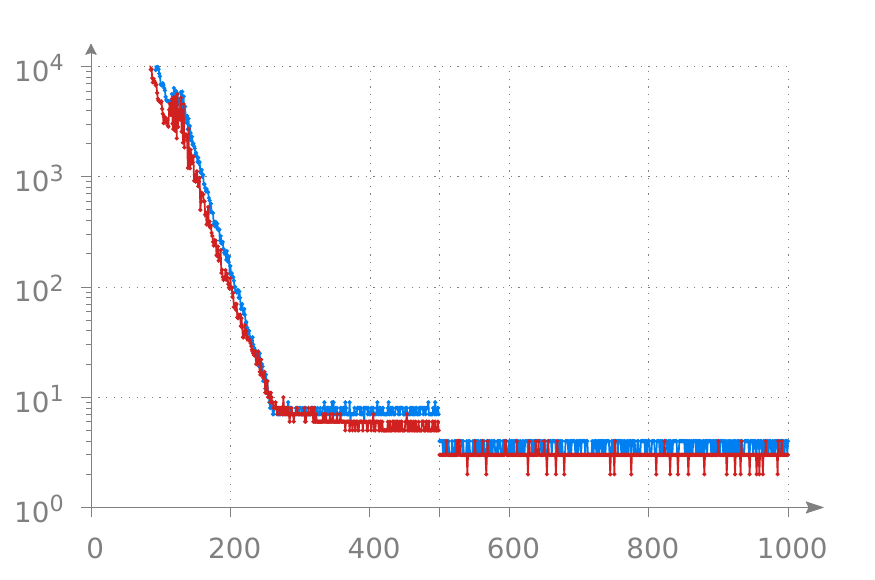} \hfill \includegraphics{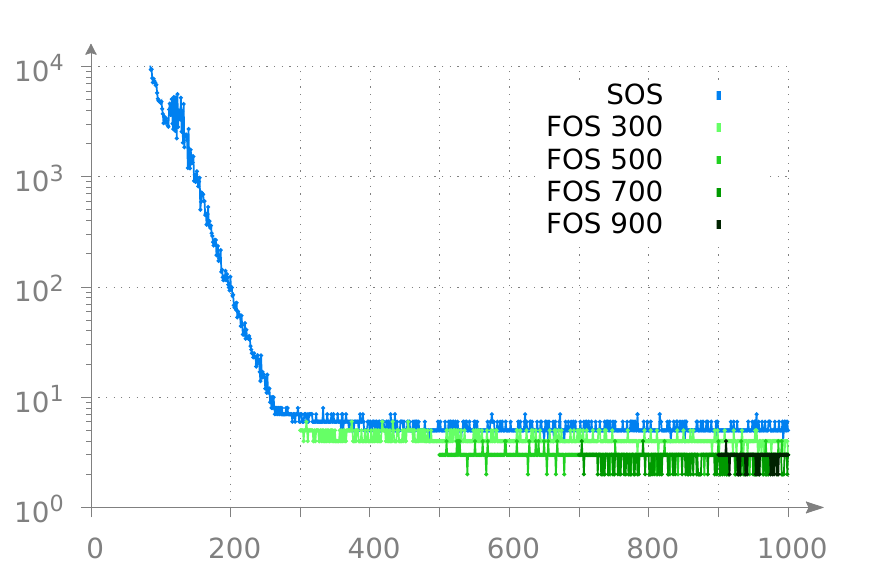}
\caption{A plot showing the effect of switching from SOS to FOS on a
two-dimensional torus of size $100\times100$. The left plot shows the maximum
load minus the average load in blue and the maximal local load difference in
red. After 500 SOS rounds the process switches to a FOS approach. In the right
plot various time steps to switch from SOS to SOS are used. All data points
show the current maximum load minus the average load.}
\label{fig:2d-torus-100x100-fos}
\end{figure*}

In Figure \ref{fig:sos} we also observe strong discontinuities of the local and
global maximum load differences which occur approximately every $1200$ to $1300$ steps.
To explain these discontinuities we visualized the load balancing process on
the two-dimensional torus in Figure \ref{fig:steps-bw}
as follows. We rendered
a raster graphic of size $1000\times1000$ pixels per round. In the graphic each
pixel represents a node of the torus such that neighboring pixels are connected
in the network and border-pixels are connected in a periodic manner. We now set
the pixels' colors to correspond to the nodes' loads, i.e., a pixel is shaded
bright if its load is close to the average load and dark otherwise. In the
visualization shown in Figure \ref{fig:steps-bw}
the initial load is placed at
the node with ID $0$, which corresponds to the top-left pixel. Since the
border-pixels \emph{wrap around}, the loads spread in circles from all four
corners, forming the \emph{wavefronts} in the graphic. Our visualizations now
indicate that the discontinuities in the local load differences and the maximum
load occur whenever these wavefronts collapse at the center of the graphic,
i.e., when the center node gets load for the first time. This is a consequence
of the second order scheme since nodes continue to push loads towards the
center pixel, even though this pixel may already have a load above average.
Note that these discontinuities also occur in the idealized scheme and for
smaller tori, see Figures \ref{fig:comparison-idealized} and
\ref{fig:2d-torus-100x100-fos}, respectively.

We furthermore rendered a video of the load balancing process (available
online, see \cite{video}) which shows the behavior of the system in an
intuitive way and thus helps understanding these discontinuities. Further
visualizations in \autoref{fig:smooth} show the
impact of the first order scheme. That is, after applying FOS steps the
rendered image becomes more \emph{smooth}, in contrast to the SOS steps where
our visualization shows  a significant amount of noise. 

\begin{figure}
\centering
\frame{\includegraphics[width=0.8\columnwidth]{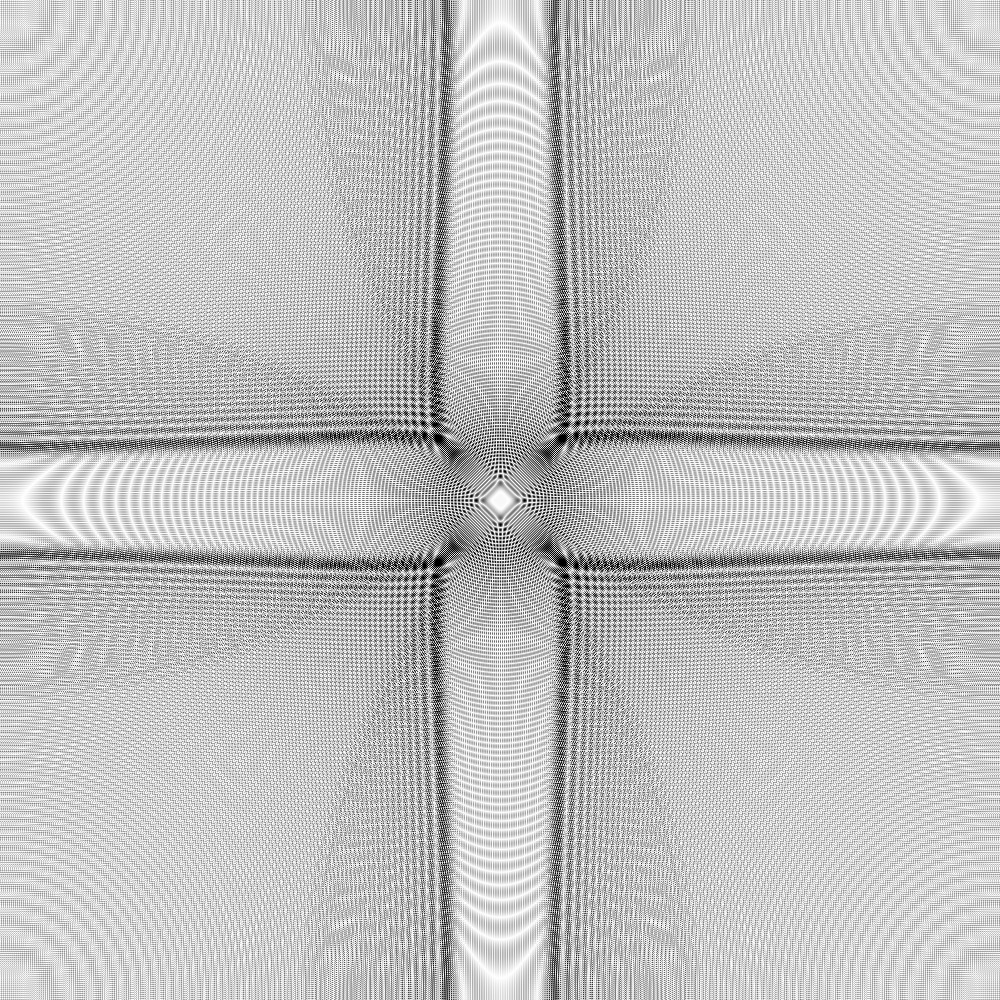}}
\caption{A visualization of the load balancing network (a two-dimensional torus
of size $1000\times1000$) after 1100 steps. Each pixel corresponds to one node
which has edges to its $4$-neighborhood and is shaded such that a light pixel
has a load close to the average load and a dark pixel a load close to either the
maximum or minimum load of the system. Further time steps are rendered in \autoref{fig:steps-bw-further-steps}}
\label{fig:steps-bw}
\end{figure}

\begin{figure*}[h]
\centering
\frame{\includegraphics[width=0.35\textwidth]{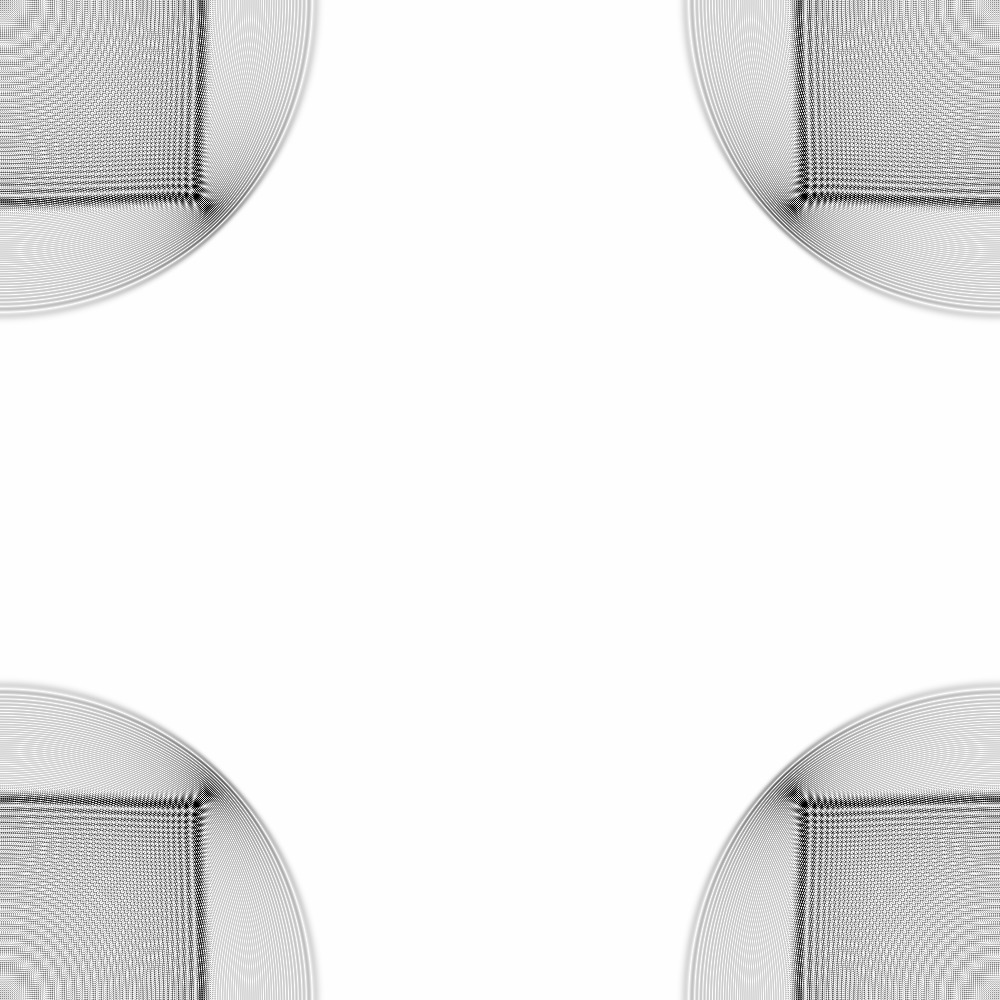}} \quad \frame{\includegraphics[width=0.35\textwidth]{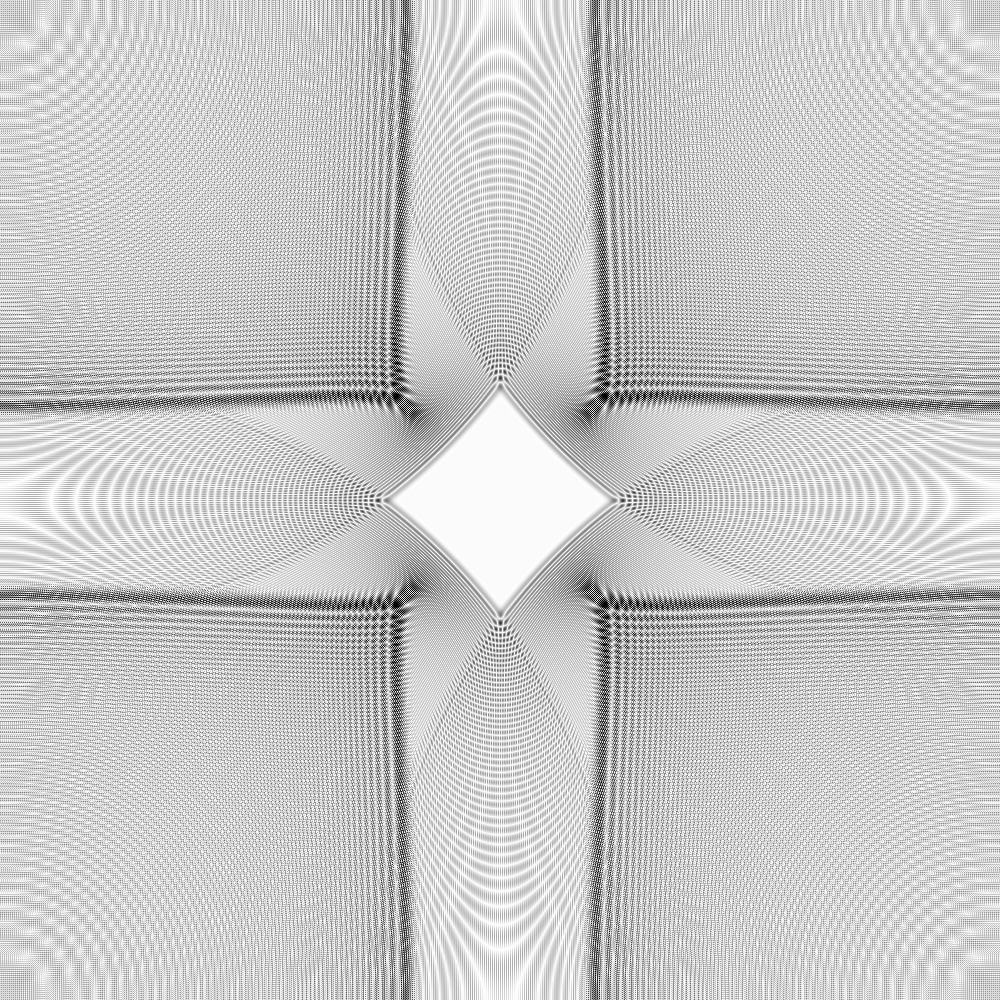}}\\[2ex]
\frame{\includegraphics[width=0.35\textwidth]{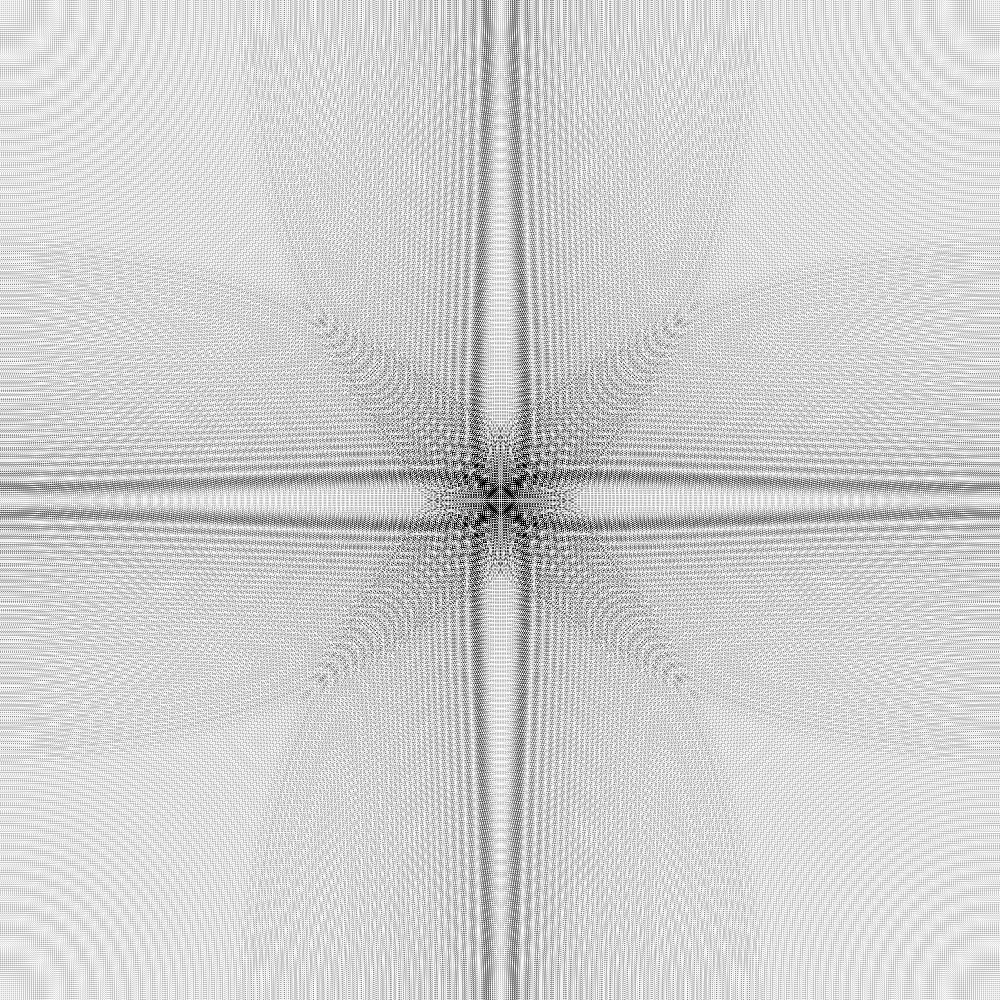}} \quad \frame{\includegraphics[width=0.35\textwidth]{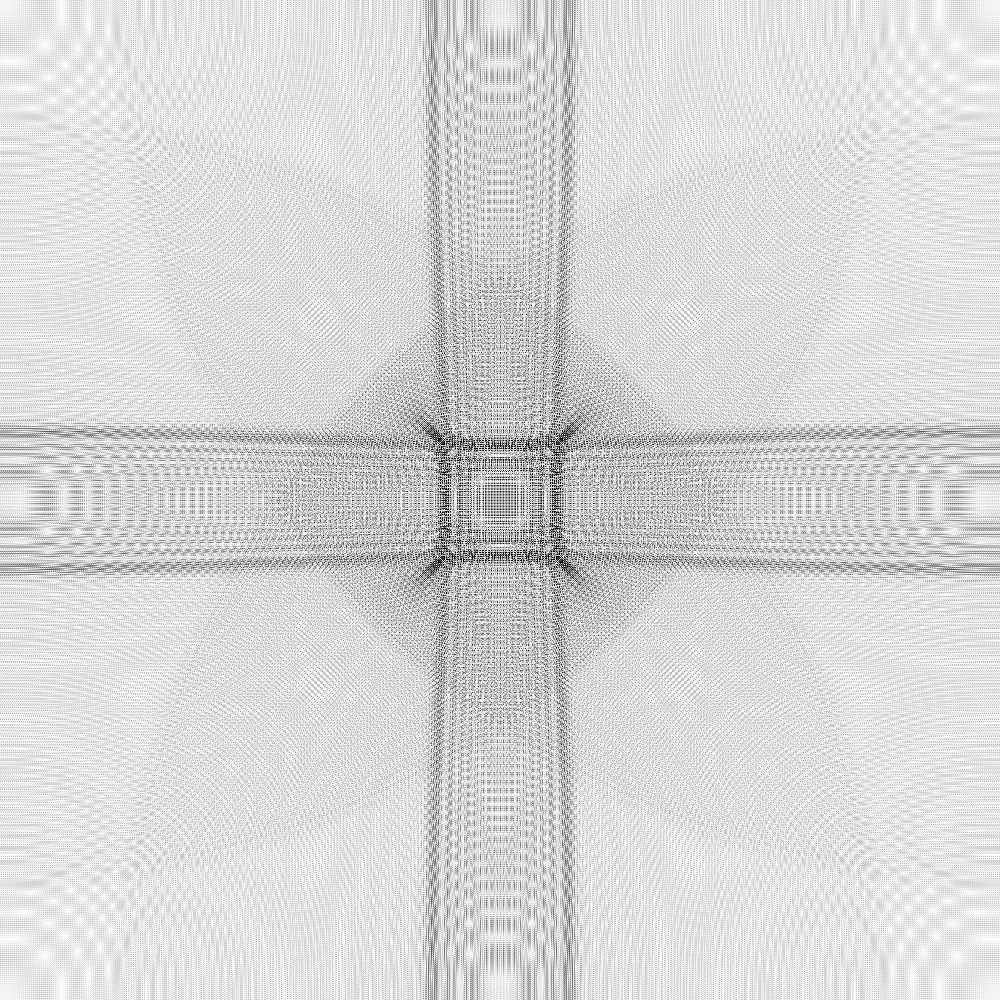}}
\caption{The figure shows the same visualization as Figure \ref{fig:steps-bw},
rendered after 500, 1000, 1200, and 1400 steps. The network is modeled as a
two-dimensional torus. Each pixel corresponds to one node which has edges to
its $4$-neighborhood. All pixels are shaded in an adaptive way, i.e., a light
gray or white pixel indicates a load close to the average load and a dark gray
or black pixel indicates a load close to either the maximum or minimum load of
the system.} \label{fig:steps-bw-further-steps}
\end{figure*}

\begin{figure*}[h]
\centering
\includegraphics[width=0.31\textwidth]{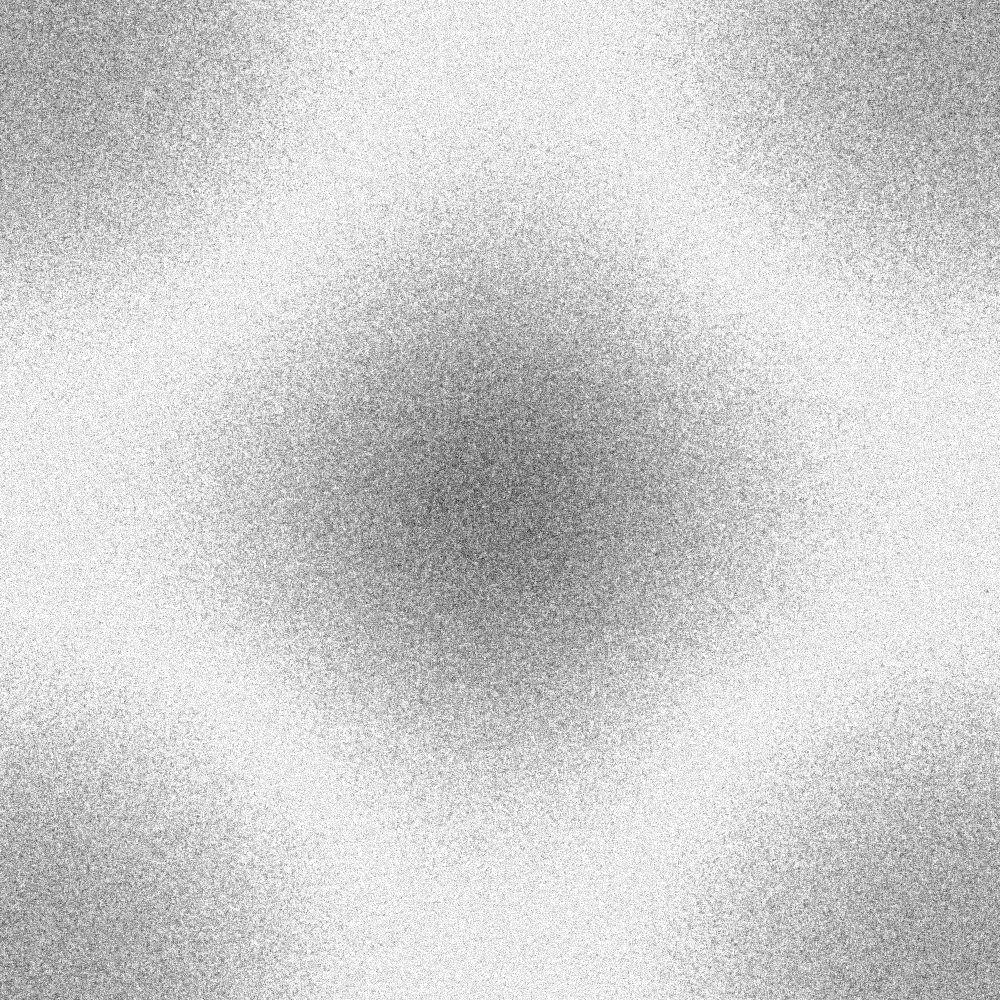} \hfill
\includegraphics[width=0.31\textwidth]{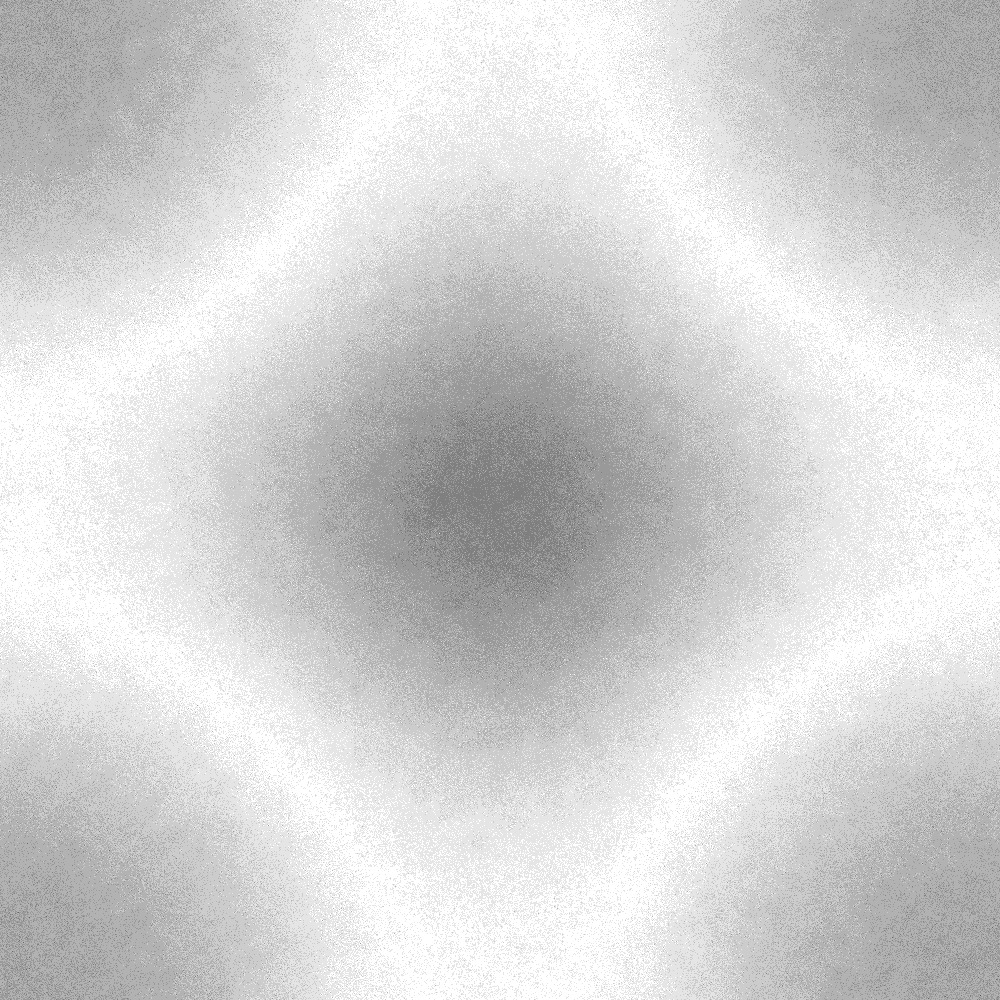} \hfill
\includegraphics[width=0.31\textwidth]{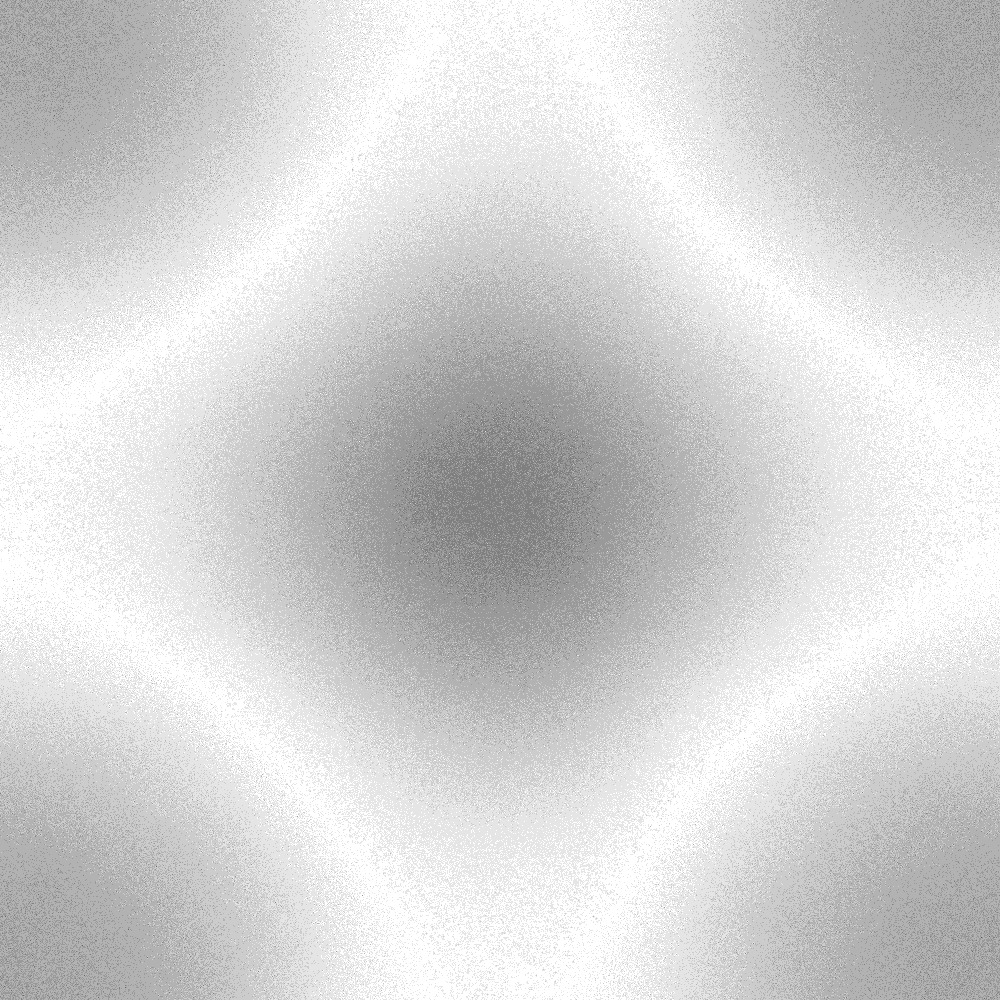}
\caption{Above figures show the same load balancing network as Figure
\ref{fig:steps-bw}. A pixel colored white indicates a node with optimal load, a
pixel colored black corresponds to a load that is more than $10$ units away
from the optimal value. Observe that in none of the above images such a load
(which exceeds the average load by more than 10 tokens) occurs. In the center
region of the left image there are several pixels which have load at least $9$,
whereas in the right image the maximum load exceeds the average load by at most $7$.
The first visualization has been rendered after 3000 SOS steps. The second
image and the third image show the same network after additional 100 and 1000
FOS steps, respectively. } \label{fig:smooth}
\end{figure*}

To gain further insights we also implemented a simulation of the idealized load
balancing procedure where loads can be split up in arbitrary small portions and
any real fraction of load can be transmitted. This simulation is based on
double precision floating point variables that represent the current load at a
node. Therefore, a quantification takes place which introduces an error. However, we observed that in our
setup the total error over all
loads is small and thus can be neglected. A comparison
of the idealized and discrete processes can
be found in Figure \ref{fig:comparison-idealized}.

\subsection{Other Networks}\label{sect:other}

For random regular graphs constructed using the configuration model
\cite{Wor99} and the hypercube we observe only a limited improvement of
SOS compared to FOS, see Figures \ref{fig:random-graph} and \ref{fig:hypercube}, respectively.  That is, the number of steps required to balance the
loads up to some additive constant is only slightly larger when using FOS
instead of SOS. For random graphs the remaining imbalance is the same for both
FOS and SOS. For the hypercube our results indicate that the remaining
imbalance using FOS is by one smaller than in the case of the SOS process.
Hence, our data only show a negligible difference between FOS and SOS in
these graphs. This can be related to the second largest eigenvalue of the diffusion
matrix, which is $(2+o(1))/\sqrt{d}$ for random graphs and $1-2/(\log{n}+1)$ for hypercubes (compared to approximately $1 - \pi^2/n$ for the torus) \cite{CDS80}. Note that the
spectral gap is also reflected in the corresponding values for
$\beta$ in \autoref{tab:graph-types}.

The random geometric graphs were generated by assigning each node
a coordinate pair in the range $[0, \sqrt{n}]^2$ uniformly at random and connecting nodes $v_i$ and
$v_j$ if and only if $d(v_i, v_j) \leq \sqrt[4]{\log{n}}$, where $d$ denotes the euclidean
distance. Remaining small isolated components were connected to the closest neighbor
in the largest component of the graph.  
Even though we observe a less pronounced potential drop in random geometric graphs, the behavior of FOS and SOS in these graphs is very similar to the behavior in the torus graphs, see Figures \ref{fig:random-geometric-graph} and \ref{fig:2d-torus-100x100}.


\begin{figure*}[H]
\includegraphics{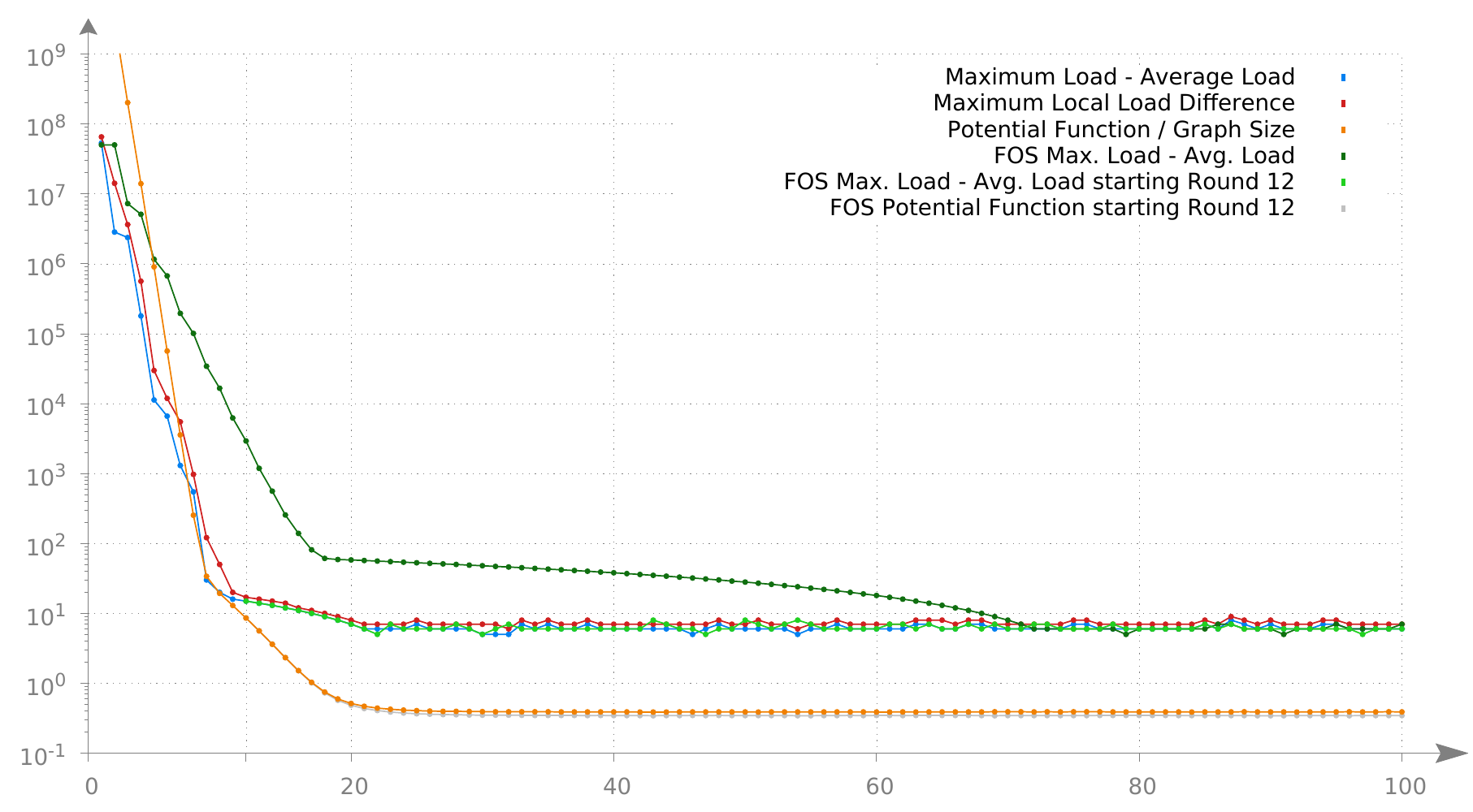}
\caption{Load balancing simulation on a random graph in the configuration model
of size $10^6$ nodes with $d = 19$.}
\label{fig:random-graph}
\end{figure*}

\begin{figure*}[H]
\includegraphics{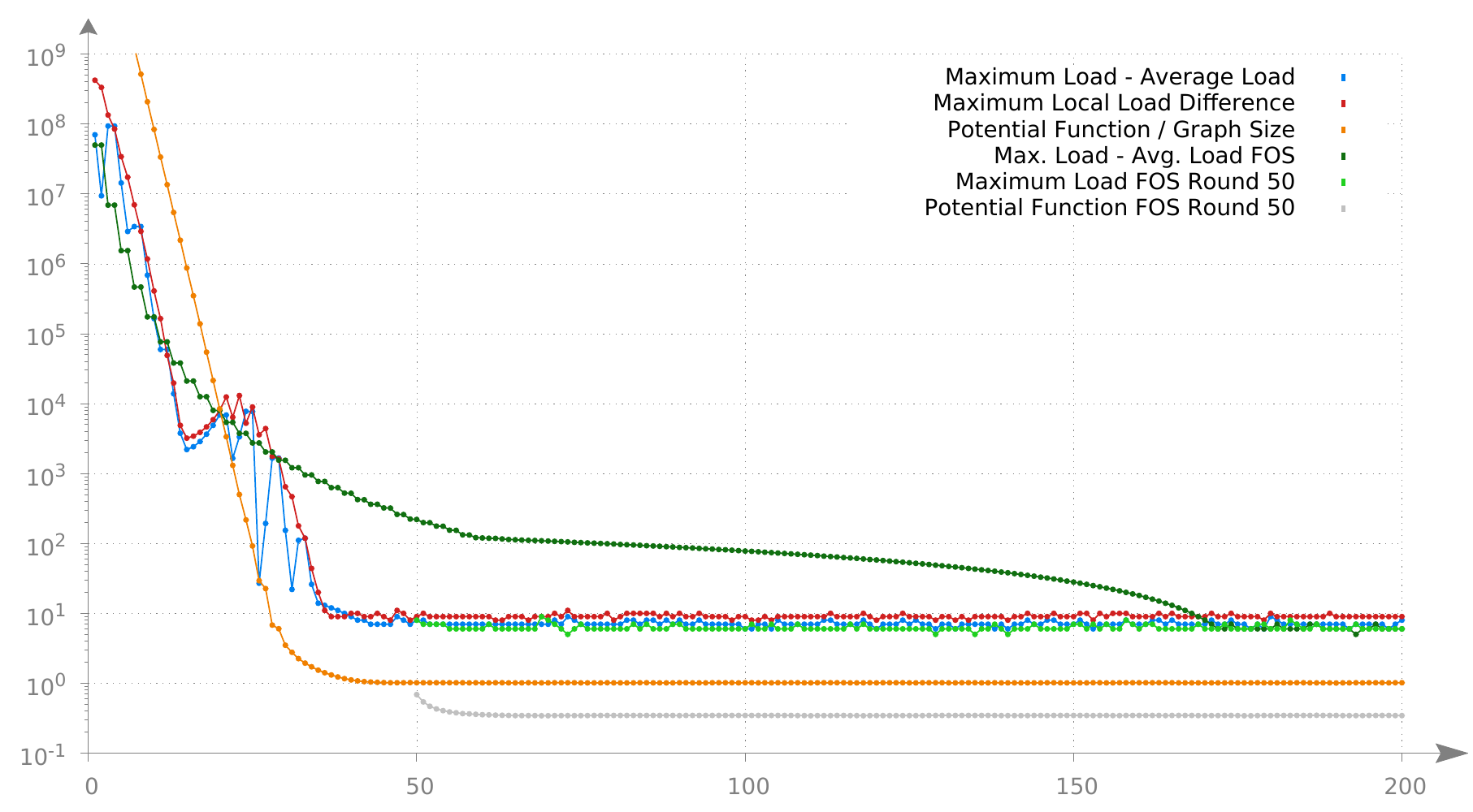}
\caption{Load balancing simulation on a hypercube with $n=2^{20}$ nodes. The
green data points show the effect of switching to FOS after 32 steps.}
\label{fig:hypercube}
\end{figure*}

\begin{figure*}[H]
\includegraphics{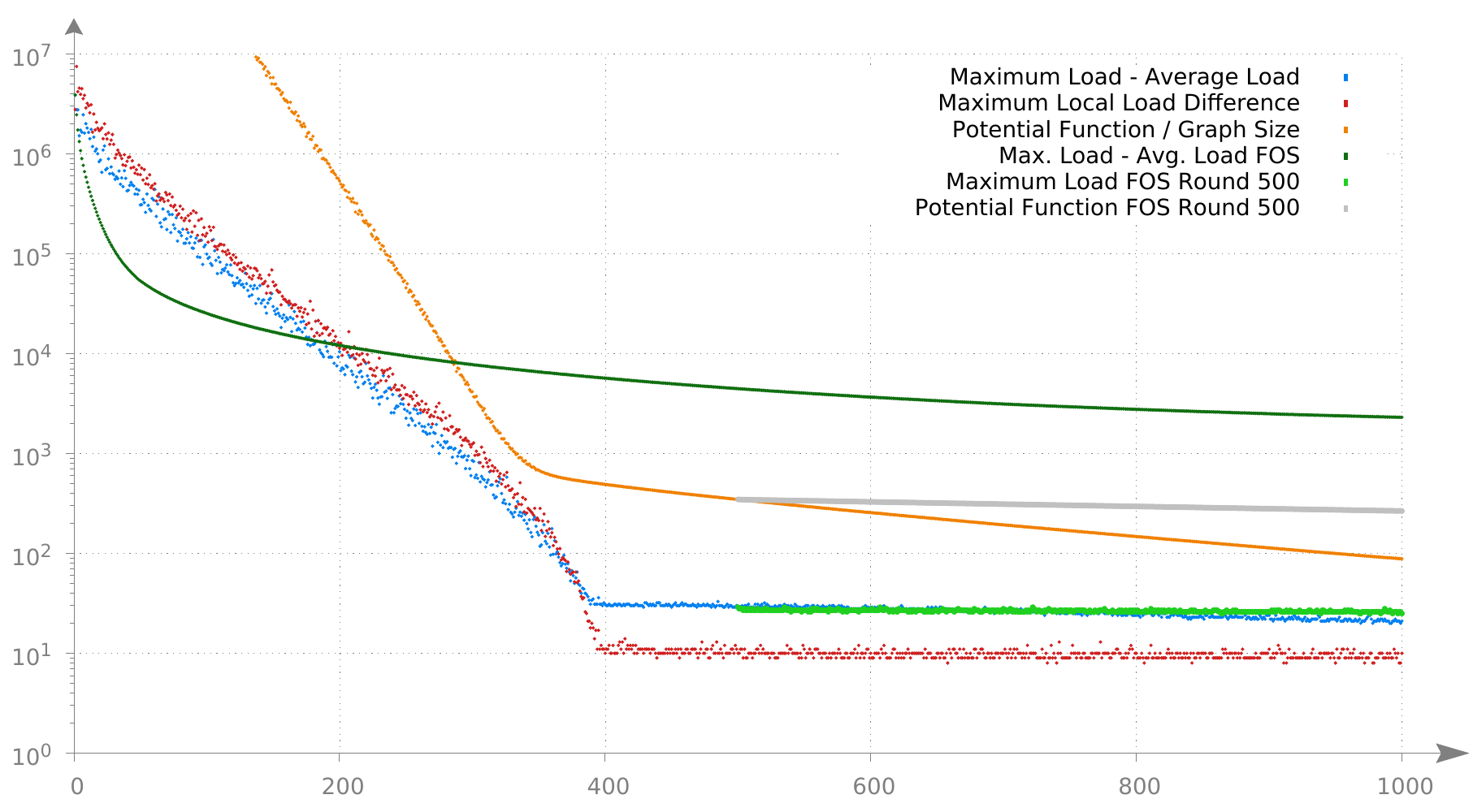}
\caption{Load balancing simulation on a random geometric graph with $10.000$
nodes in $[0, \sqrt{n}]^2$ with connectivity radius $\sqrt{\log{n}}$.}
\label{fig:random-geometric-graph}
\end{figure*}

\begin{figure*}[H]
\centering
\includegraphics{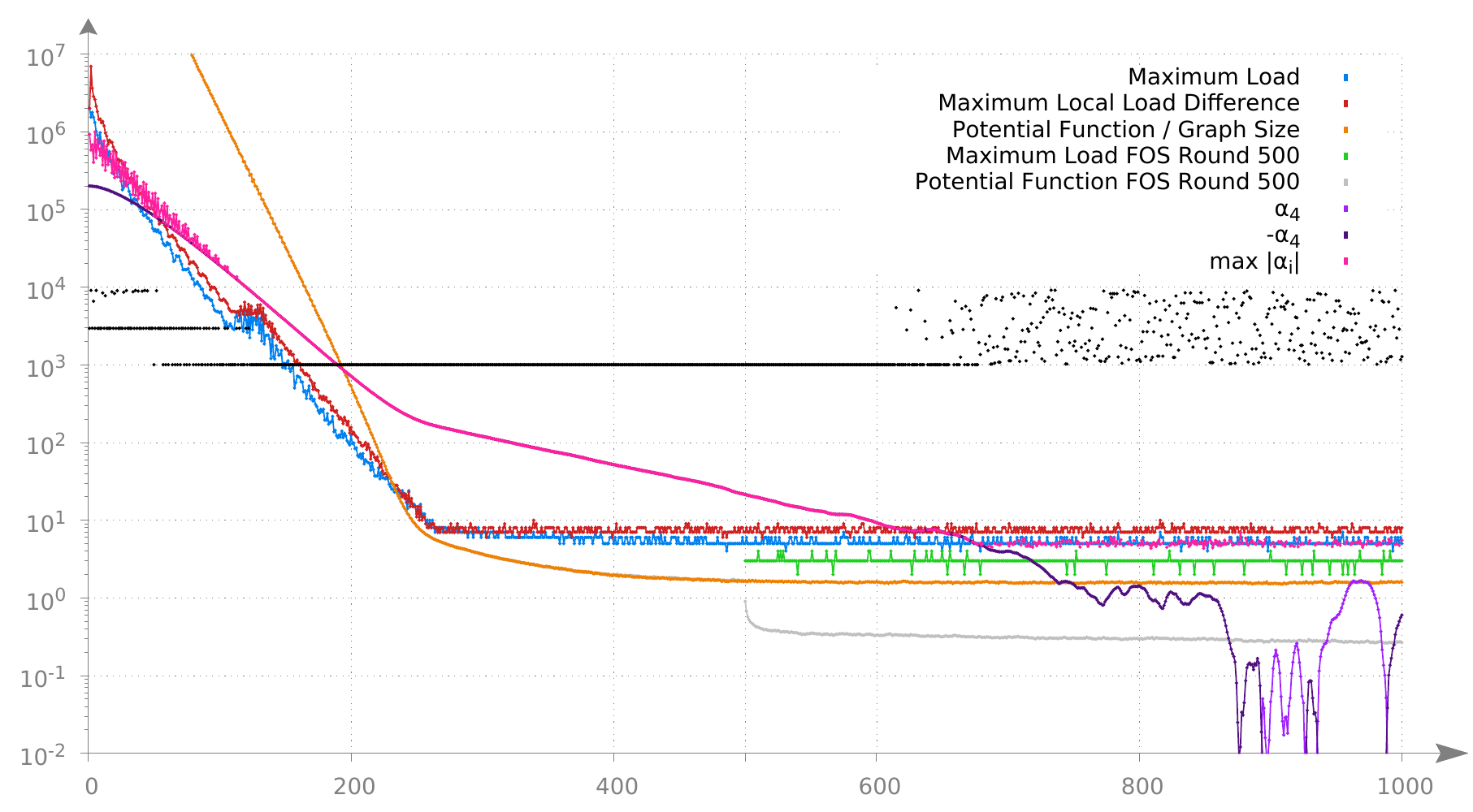}
\caption{Load balancing simulation on a two-dimensional torus of size
$100\times100$. The purple line shows the maximum coefficient
$\max\{|\alpha_i|\}$ for the impact of the eigenvectors on the load. This
coefficient is $-\alpha_4$, starting approximately in round 100 and up to
approximately round 700. The black dots also shown in this plot in the range
$[10^3, 10^4]$ represent the leading coefficient, where $\alpha_1$ is plotted
with a value of $10^3$ and $\alpha_n$ is plotted with a value of $10^4$, with a
linear scale between them. Observe that after approximately 700 rounds no
single eigenvector can be identified that has a leading impact on the load.}
\label{fig:2d-torus-100x100}
\end{figure*}

%% file: conclusion.tex
\section{Conclusion}

In this paper we analyzed a broad class of discrete diffusion type algorithms
by comparing them to their continuous counterparts. Furthermore, we studied the
problem of negative load in second order schemes and presented a bound for the
initial minimum load in the network in order to avoid negative load during the
execution of the algorithm.

Our analyses seem to provide bounds for the negative load and for the arbitrary
rounding of SOS which leave room for improvement. However, in order to tighten
these results, one needs some different analytic techniques. Therefore, we
think that any improvement would be an interesting contribution to the field of
second order diffusion schemes in particular and load balancing algorithms in
general.